\documentclass[11pt,a4]{amsart}
\usepackage{amssymb,epsfig}
\usepackage{amsmath}
\usepackage{color,layout}
\newtheorem{theorem}{Theorem}[section]

\newtheorem{lemma}[theorem]{Lemma}
\newtheorem{proposition}[theorem]{Proposition}

\newtheorem{remark}[theorem]{Remark}
\newtheorem{corollary}[theorem]{Corollary}
\newtheorem{example}[theorem]{Example}
\usepackage{amsmath,amsfonts,epsfig,latexsym}
\def\square{\hbox{\vrule\vbox{\hrule\phantom{o}\hrule}\vrule}}

{\vspace{2ex}\noindent{\sl Proof: }}{{\null\hfill\square\bigskip}}

\def\re{{\rm Re}}
\def\im{{\rm Im}}
\def\ord{{\mathcal O}}
\def\W{{\mathcal W}}
\def\R{\mathbb {R}}
\def\H{\mathcal {H}}
\def\C{\mathbb {C}}
\def\N{\mathbb {N}}
\def\eq#1{(\ref{#1})}
\def\e{\varepsilon}
\def\l{\lambda}

\title[Semiclassical WKB]{Semiclassical WKB problem for the non-self-adjoint Dirac operator with analytic potential} 
\author[S.Fujiie]{Setsuro Fujii\'e}
\address{Department of Mathematical Sciences, Ritsumeikan University,
1-1-1 Nojihigashi, Kusatsu, Shiga, 525-8577, Japan}

\author[S. Kamvissis]{Spyridon Kamvissis}
\address{Department of Pure and Applied Mathematics,
University of Crete, GR--700 13 Voutes Campus, Greece, and Institute
of Applied and Computational Mathematics, FORTH, GR--711 10
Voutes Campus, Greece}

\begin{document}
\setcounter{section}{0}
\maketitle

\abstract
In this paper we examine the semiclassical behaviour of the scattering data of a non-self-adjoint Dirac operator with analytic potential decaying at infinity.
In particular, employing the exact WKB method, we
provide the complete rigorous uniform semiclassical analysis of the reflection coefficient
and the Bohr-Sommerfeld condition for the location of the eigenvalues. 
Our analysis has some interesting consequences concerning
the focusing cubic NLS equation, in view of the well-known fact discovered by Zakharov and Shabat that 
the spectral analysis of the Dirac operator is the basis of the solution of the NLS equation via
inverse scattering theory.
\endabstract

\section{Introduction: Motivation}

In the last twenty  years or so  
the analysis of the semiclassical behaviour of the focusing NLS equation has been rigorously achieved 
(and also numerically supported and clarified) for 
a certain class of real analytic decaying initial data (\cite{kmm}, \cite{kr}, \cite{k1}, \cite{k2}, \cite{lm}, \cite{bt}).

The problem is as follows:
consider the semiclassical limit ($\epsilon \to 0$)
of the solution to the initial value problem of the one-dimensional nonlinear Schr\"odinger equation: 
\begin{equation}
\left\{
\begin{array}{l}
i\epsilon\partial_t\psi +
\frac{\epsilon^2}{2}\partial_x^2\psi + |\psi|^2\psi = 0, \\[8pt]
\psi(x,0)=A(x), 
\end{array}
\right.
\end{equation}
We assume here that $A(x)$ is a {\it real analytic} integrable function, and moreover that it is a
positive  ``bell-shaped" function;
in other words 
\begin{equation}
A (x) >0, \quad
A (-x) = A (x),\\
\end{equation}
and it has  one single non-degenerate maximum at 0, say $A_0$,
\begin{equation}
A (0)=A_0>0,\quad xA'(x) < 0,\quad
A''(0) < 0, 
\end{equation}

Suppose now that
we replace the initial data by the so-called ``soliton ensembles" data  which are defined by replacing
the scattering data for $\psi (x,0)= A(x)$ with  their {\it formal}
WKB-approximation: we set the reflection coefficient of the associated Dirac operator (see section 2) to be identically zero and 
replace the actual eigenvalues  by their Bohr-Sommerfeld approximation (see  section 5).
In other words we replace the initial data by a {\it new} set of data which is now depending on $\epsilon$.
Suppose that we solve the focusing NLS equation under this new set of initial data. Then we have 
the following.

Let $x_0, t_0$ be any given point ($x_0 \in \Bbb R, t_0 > 0$).
The solution $\psi(x,t)$  is asymptotically ($\epsilon \to 0$)
described  (locally) as a slowly
modulated $G+1$ phase wavetrain.  Setting $x=x_0+\epsilon  \hat{x}$
and $t=t_0+\epsilon  \hat{t}$,
so that $x_0, t_0$ are ``slow" variables
while $\hat{x}, \hat{t}$ are ``fast" variables,
there exist
parameters
$a,  U = (U_0, U_1, .... , U_G)^T,$
$k =(k_0, k_1, ......, k_G)^T,$
$w =(w_0, w_1, ....., w_G)^T, $
$Y =(Y_0, Y_1, ........., Y_G)^T,$
$Z =( Z_0, Z_1, ...... , Z_G)^T $
depending on the slow variables
$x_0$ and $t_0$  (but not  $\hat{x}, \hat{t}$)
such that
$\psi(x,t)= \psi(x=x_0+ \epsilon  \hat{x}, t=t_0+\epsilon  \hat{t})$ has the following 
leading order asymptotics as $\epsilon \to 0$: 

\begin{equation}
\label{asymptotics}
\psi(x,t) =
a(x_0, t_0) e^{iU_0(x_0, t_0)/\epsilon}
e^{i(k_0(x_0, t_0) \hat{x}-w_0(x_0, t_0) \hat{t})} \hspace{2.5cm}
$$
$$
\hspace{1.3cm}\cdot  \frac{\Theta(  Y(x_0, t_0)+
i  U(x_0, t_0)/\epsilon +
i(  k(x_0, t_0) \hat{x}-  w(x_0, t_0)\hat{t}))}
{
\Theta(  Z(x_0, t_0)+
i  U(x_0, t_0)/\epsilon +
i(  k(x_0, t_0)\hat{x}-  w(x_0, t_0)\hat{t}))} (1+o(1)).
\end{equation}

All parameters can be defined in terms of an underlying
Riemann surface $X$ which depends solely on $x_0, t_0$.   The moduli of $X$ vary slowly with $x, t$, i.e.
they depend on  $x_0, t_0$ but not on
$\epsilon, \hat{x}, \hat{t}$. 
$\Theta$ is the G-dimensional Jacobi theta function associated with $X$.
The genus of $X$ can vary with $x_0, t_0$. In fact, the $x,t$-plane is divided into open regions in each of which
G is constant. On the boundaries of such regions (sometimes called ``caustics"; they are unions of analytic arcs), 
some degeneracies appear in the mathematical analysis (we may have ``pinching" of the surfaces $X$ for example) and
interesting physical phenomena can appear (like the famous Peregrine rogue wave \cite{bt}).
The above formulae give pointwise asymptotics, which are in fact uniform
in compact (x,t)-sets not containing points on the caustics.
For the exact formulae for the parameters as well as the definition  of the theta functions we refer to \cite{kmm}  or \cite{kr}.
For an analysis of the somewhat more delicate behaviour (especially for higher order terms in $\epsilon$)
near the first caustic see \cite{bt}.

The above result is interesting but somewhat unsatisfactory.
The reason, of course, is that the initial data is substituted by 
the soliton ensembles data. A rigorous justification of this substitution 
requires rigorous semiclassical asympotics
for the spectral data of the Dirac operator that is associated to the focusing NLS equation
(see the next section).
Our main aim in this paper is to show 
how the powerful ``exact WKB method" can be used to provide the necessary
rigorous asymptotic results.

The question of the semiclassical approximation of the scattering data has a deeper significance 
in view of the instability of the problem which appers in many levels.
In fact even in the non-semiclassical regime, the {\it focusing} NLS is the main model for the so-called
``modulational instability" (\cite{bf}, \cite{bm}), although for positive fixed $ \epsilon $ the initial value problem is well-posed.

Semiclassically the  instabilities become more pronounced.
One way to see this  is related to the underlying ellipticity
of the formal semiclassical limit. 
To be more specific, consider the well-known Madelung transformation \cite{m}.

\begin{equation}
\rho = |\psi|^2, \\~~~~~~~~~~ \mu = \epsilon \im\,(\bar \psi \psi_x).
\end{equation}
Then the initial value problem becomes
\begin{equation}
\rho_t  +\mu_x = 0,\\ ~~~~~~~~~\mu_t + ({\mu^2 \over \rho} + {\rho^2 \over 2})_x = {\epsilon^2 \over 4} \partial_x(\rho (\log \rho)_{xx}),
\end{equation}
with initial data $\rho(x,0) =  (\psi(x,0))^2$ and $\mu(x,0)=0$.

The formal limit as $\epsilon \to 0$ is
\begin{equation}
\rho_t  +\mu_x = 0,\\ ~~~~~~~~~\mu_t + ({\mu^2 \over \rho} + {\rho^2 \over 2})_x = 0,
\end{equation}
with initial data $\rho(x,0) = (\psi(x,0))^2$ and $\mu(x,0)=0$.

This is an initial value problem for an elliptic system of equations and so one expects that small perturbations of the initial data
(independent of $\epsilon$) can lead to large changes in the solution, at any given time.

Another appearance of instabilities appears at the spectral analysis of the related non-self-adjoint Dirac operator
(see the next section). Instability appears also at the related equilibrium measure problem 
(see section 6 and the appendix),  the related Whitham equations (they are also elliptic) 
 and even in the numerical studies of the problem.

As already stated, the semiclassical approximation of the scattering data results in
 small changes of the initial data $that ~depend~on~\epsilon$.
It is a priori unclear whether they can have a significant effect in the semiclassical asymptotics of the solution at a given time. 
Our aim is to prove that, at least for  these particular initial data, they do not.

In simpler problems like the real KdV equation, 
or the defocusing nonlinear Schr\"odinger equation, one can make use of the underlying
hyperbolicity of the formal limit to prove, a posteriori, that the formal semiclasssical WKB analysis of the scattering data 
is justified. In the focusing nonlinear Schr\"odinger equation, 
we need more delicate tools, provided by the exact WKB method.  
The exact WKB method was first developed for the Schr\"odinger operator,  but here we apply it
to the Dirac operator that is associated to the focusing NLS equation. The method goes back to works of Ecalle 
\cite{e} and Voros \cite{v} but here we argue along the lines of the papers of G\'erard-Grigis \cite{gg} 
and Fujii\'e-Lasser-N\'ed\'elec \cite{fln}.
Rather than relying on the usual formal WKB method which relies on asymptotic series that are in general divergent,
we use a ``resummation" of the series and in fact construct $exact$ solutions in terms of $convergent$ series, 
thus resolving a problem of ``asymptotics beyond all orders".
 
We begin by addressing the issue of the reflection coefficient. 
The exact WKB method is employed to prove that it is exponentially small
away from the point 0, in the spectral plane. 
Similarly, we  give a rigorous justification of the Bohr-Sommerfeld asymptotic conditions
for the locations of the eigenvalues.  Our main  assumptions on the initial data, i.e. the potential
of the Dirac operator are two: analyticity near the real line  and a mild decay estimate.
Some extra technical assumptions are needed for the analysis of the scattering data near $0$.
These assumptions are often cumbersome to check.
Still a wide enough open class of ``potentials" $A(x)$ satisfies these assumptions, including
positive  bell-shaped  rational functions in $L^1$ and also
exponential functions like $(\cosh x)^{-1}$, $e^{-x^2}$ (see Example \ref{pol}, \ref{exp}).

The plan of this paper is the following. In section 2 we state the exact  assumptions on the potential and present the results
on the rigorous WKB approximation. In section 3 the exact WKB method is presented. 
In section 4, it is applied to the reflection coefficient
of  the Dirac operator. In section 5, the eigenvalues are considered. 
In section 6, we present  the application of the WKB results to the focusing NLS problem
and  we explain  how the analysis of \cite{kmm} needs to be  modified in view of these results. 

In the  appendix, we first present the Riemann-Hilbert problem for the focusing NLS equation.
We then give a rudimentary description of the change of variables needed to asymptotically deform the
given Riemann-Hilbert problem  into a ``model" problem that can be explicitly solved.
Finally, we present a discussion of the results
of \cite{k1}, which concern the possible obstacle of the non-analyticity of the spectral density of eigenvalues.

\section{Assumptions and results}
We study the semiclassical asymptotics of the reflection coefficient and the eigenvalue distribution of the Dirac operator
$$
L:=\left (
\begin{array}{cc}
-\frac \epsilon i\frac{d}{dx} & -iA(x) \\[8pt]
-iA(x) & \frac \epsilon i\frac{d}{dx}
\end{array}
\right ).
$$
Here, $\e>0$ is the semiclassical small parameter and $A(x)$ is a function satisfying

\begin{description}
\item[(A1)]
$A(x)$ is a real positive smooth function on $\R$, and extends analytically to the complex domain
$$
D_0=D(\rho_0,  \theta_0):=\{x\in \C; |\im \,x|<\max (\rho_0, (\tan\theta_0 )|\re \,x|)\}
$$
for some positive $\rho_0$ and $\theta_0$.
Moreover, there exists a positive $\tau$ such that, as $x\to\infty$ in $D_0$,
$$
A(x)=\ord (|x|^{-1-\tau}).
$$
\end{description}

Under this condition, it is known that the spectrum of the non-self-adjoint operator $L$ consists of the continuous spectrum  $\R$ and 
a finite number of eigenvalues coming in complex conjugate pairs (and which are close to $i[-A_0,A_0]$, where $A_0:=\max_{x\in \R}A(x)>0$
 when $\epsilon$ is small).

We will first study the asymptotic behavior of the reflection coefficient $R(\l,\e)$ for 
$\l\in\R\setminus \{0\}$.

First we have the following result for the reflection coefficient for $\l\in\R$ away from $0$.

\begin{theorem}
\label{ref1}
Assume (A1). Then, for any $\delta>0$, there exists $\sigma>0$ independent of $\e$ such that
$$
|R(\l,\e)|=\ord (e^{-\sigma/\e}),
$$
as $\e\to 0$ uniformly for $\l\in (-\infty, \delta]\cup [\delta,\infty)$.
\end{theorem}

For the eigenvalues, we assume moreover that $A(x)$ is a ``bell-shaped" function:

\begin{description}
\item[(A2)]
$A(x)=A(-x)$ and $A'(x)<0$ for $x>0$.
\item[(A3)]
$A''(0)<0$.
\end{description}

We will also study the accuracy of the quantization condition of the eigenvalues on $i[-A(0),A(0)]$ as  $\epsilon \to 0$.

Let $\lambda=i\mu$ with $0<\mu<A_0=A(0)$. The assumption (A2) implies that there exists  a unique positive $x^*(\mu)$ such that there are exactly two real numbers $x^*(\mu)$ and $-x^*(\mu)$  which
satisfy $A(x)=\mu$. Define an action integral 
\begin{equation}
\label{action}
S(\mu)=\int_{-x^*(\mu)}^{x^*(\mu)} \sqrt{A(x)^{2}-\mu^2}dx.
\end{equation}
\begin{theorem}
\label{ev1}
Assume (A1), (A2) and (A3). 
Then there exists a function $m(\mu,\e)$ with asymptotic behavior
$$
m(\mu,\e)=-1+\ord (\e)
$$
as $\e\to 0$ uniformly in any closed interval $I\subset (0,A(0)]$
such that $\l=i\mu$ where $\mu\in I$ is an eigenvalue of $L$ if and only if
\begin{equation}
\label{BS}
m(\mu,\e)e^{2iS(\mu)/\e}=1.
\end{equation}
\end{theorem}

\begin{remark}
Klaus and Shaw proved in \cite{ks} that all the eigenvalues are simple and  purely imaginary under the bell-shaped
conditions (A1), (A2).
Recently Hirota and Wittsten refined Theorem \ref{ev1} to show that the
 eigenvalues  are still pure imaginary even if we only  impose an ``energy-local" bell-shaped condition
for small  enough $\epsilon$ (see \cite{hs}).
\end{remark}

Now let us focus on the asymptotic behavior of the functions $R(\lambda,\e)$ and $m(\mu,\e)$ when
$\l>0$ or $\mu>0$ tends to $0$ together with $\e$. In such a case, we need a more precise assumption on 
the asymptotic
behavior of the potential $A(x)$ as $|x|\to \infty$ in $D$.

We define a function
\begin{equation}
\label{zx}
z(x)=i\int_0^x\sqrt{A(t)^2+\lambda^2}dt,
\end{equation}
where we take the branch of the square root such that it is positive at $t=0$.
This function is well-defined and  holomorphic at least near the origin $x=0$. It is extended in $D_0$ except at the turning points, i.e.
the zeros of $A(t)^2+\lambda^2$, around which it is multi-valued.

We first consider the case where $\l>0$ is small. 
In this case, there is no turning point on the real axis, and the image  of the real axis by the map $x\mapsto z(x)$ is the imaginary axis.
Let $F(a)$ be the cone-like set
$$
F(a)=\{z\in \C;|\re z|<a|\im z|\}
$$
for $a>0$.
We assume
\begin{description}
\item[(A4)]
For any $\l>0$ small, there exist positive constants $\rho(\l)$, $\theta(\l)$ and $a(\lambda)$ such that
$D(\rho(\l), \theta(\l))$ contains no turning point and its image by the map $x\mapsto z(x)$
includes $F(a(\l))$.
\end{description}

\begin{theorem}
\label{ref2}
Assume (A1), (A2) and (A4). Then there exists a positive constant $c$ such that
$$
R(\l,\e)=\ord \left (e^{-ca(\l)/\e}\right ),
$$
as $\e\to +0$ and $\l\to +0$ with $\frac\e{a(\l)}\to 0$.
\end{theorem}

Next we consider the case where $\lambda=i\mu$ and $\mu>0$ is small. 
In this case, there are exactly two turning points $x^*(\mu)$ and $-x^*(\mu)$ on the real axis.
By the map $x\mapsto z(x)$, the real interval $(-x^*(\mu), x^*(\mu))$ is sent to the imaginary interval $(-z(x^*(\mu)),z(x^*(\mu)))$,
and the half line $(x^*(\mu),\infty)$ (resp. $(-\infty, -x^*(\mu))$ is sent to the half line $z(x^*(\mu))+\R_+$ (resp. $-z(x^*(\mu))+\R_+$) when
the square root in \eq{zx} is continued from $(-x^*(\mu),x^*(\mu))$ to $(x^*(\mu),\infty)$ (resp. $(-\infty, -x^*(\mu))$
passing through the upper half plane around the turning point $x^*(\mu)$ (resp. $-x^*(\mu)$).
Notice that, as $\mu\to 0$, one has $x^*(\mu)\to +\infty$ and 
$$
| z(x^*(\mu))|\to \int_0^{+\infty}A(x)dx=:z_\infty,
$$
which is a positive finite number.
Let $G(b)$ be the complex subdomain of $\C_z$ defined by
$$
\begin{array}{ll}
G(b)=&\{-b<\re z<0, |\im z|<|z(x^*(\mu))|+b)\} \\[8pt]
&\cup\{ \re z\ge 0,|z(x^*(\mu))|<|\im z|<|z(x^*(\mu))|+b\}
\end{array}
$$
for $b>0$.
We assume
\begin{description}
\item[(A5)]
For any $\mu>0$ small, there exist positive constants $\rho(\mu)$,  $\theta(\mu)$ and $b(\mu)$ such that
$D(\rho(\mu), \theta(\mu))\cap \{z\in\C;\im z>0\}$ contains no turning point and its image by the map $x\mapsto z(x)$
includes $G(b(\mu))$.

\end{description}




\begin{theorem}
\label{ev2}
Assume (A1), (A2), (A3) and (A5). Then
there exists a function $m(\mu,\e)$ with asymptotic behavior
$$
m(\mu,\e)=-1+\ord\left (\frac\e{b(\mu)}\right ) 
$$
as $\e\to +0$ and $\mu\to +0$ with  $\frac\e{b(\mu)}\to 0$,
such that $\l=i\mu$ is an eigenvalue of $L$ if and only if \eqref{BS} holds.
\end{theorem}

\begin{figure}[htbp]
\begin{center}
\includegraphics[width=90mm]{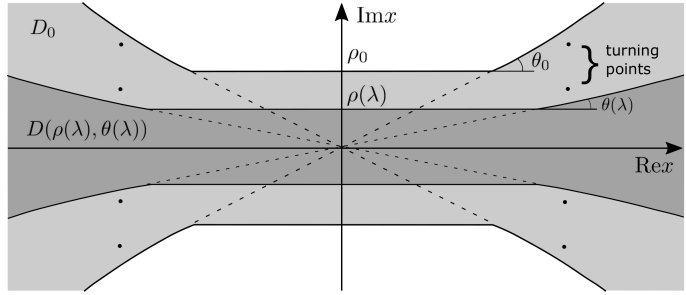}
\end{center}
\caption{The domains $D_0$ and $D(\rho(\l),\theta(\l))$}
\end{figure}

\begin{figure}[htbp]
\begin{center}
\includegraphics[width=90mm]{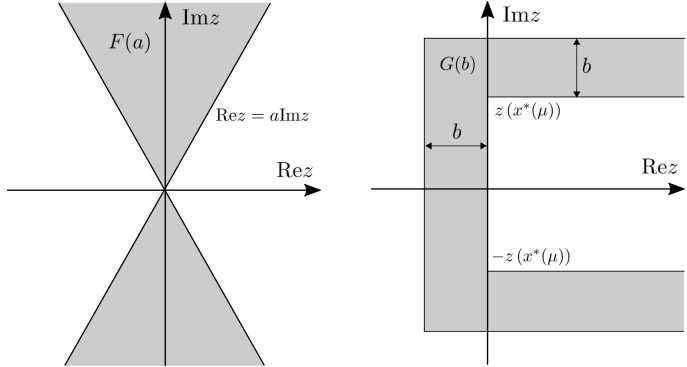}
\end{center}
\caption{The domains $F(a)$ and $G(b)$}
\end{figure}

\begin{example}
\label{pol}
Suppose $A(x)$ satifies (A1), (A2) and 
$$
A(x)={C}x^{-d}+r(x) \text{ for } x>1,
$$ 
with $d>1$, {$C>0$} and $r(x)=o(|x|^{-d-1})$, $r'(x)=o(|x|^{-d-2})$ as $\re x\to \infty$ in $D_0$.
Then, one can take $a(\lambda)=c\lambda$ and  $b(\mu)=c\mu^{1+\frac 1{2d}}$ for some positive constant $c$.
\end{example}
\begin{proof}
For simplicity, $C$ is assumed to be 1 below.
Let $\l>0$ small. Since $r(x)=o(|x|^{-d-1})$, we see by Rouch\'e's theorem that the turning points in this domain are $\pm e^{(k+1/2)\pi i/d}\l^{-1/d}+o(1)$ 
with integers $k$ satisfying 
$|k+1/2|\pi <m\theta_0$  and  the nearest turning points to the real axis are
$\pm e^{\pm \pi i/(2d)}\l^{-1/d}+o(1)$.
Hence the domain $D(\rho,\theta)$ contains no turning point for small enough $\l$-independent $\rho$ and any $\theta $ smaller than $\pi/(2d)$.
Its image  $z(D(\rho,\theta))$ includes the domain $F(a(\l))$ with $a(\l)=c\l$  for some poisitive constant $c$.
In fact, for $x\in D(\rho_1,\theta_1)$ with small enough $\rho_1, \theta_1$, 
$$
\begin{array}{rl}
|\re \,z(x)|=&\l|\im\, x|\int_0^1\re\,\sqrt{1+\l^{-2}A(\re\, x+is\im\,x)^2}ds \\[8pt]
\ge &\frac 12\l|\im\, x|,
\end{array}
$$
$$
\begin{array}{rl}
|\im\,z(x)|=&|\re \,x\int_0^{1}\sqrt{A(s\re \,x)^2+\l^2}dt \\[8pt]
&-\l\im\, x\int_0^1\im\,\sqrt{1+\l^{-2}A(\re\, x+is\im\,x)^2}ds|  \\[8pt]
& \le 2A_0 |\re\, x|.
\end{array}
$$
Hence  $F(c\l)\subset z(D(\rho_1,\theta_1))$  for $c=\frac{|\tan\theta_1|}{4A_0}$.

For $\l=i\mu$ with $\mu>0$ small,
the turning points in this domain are 
$\pm e^{k\pi i/d}\mu^{-1/d}+o(1)$ with integers $k$ satisfying 
$|k|\pi <d\theta_0$. In particular $x^*(\mu)=\mu^{-1/d}+o(1)$, and
the nearest turning points  to the real axis (apart from the real ones $\pm x^*(\mu)$) are $\pm e^{\pm\pi i/d}\mu^{-1/d}+o(1)$.
Hence $D(\rho(\mu), \theta(\mu))\cap \{z\in\C;\im z>0\}$  is turning point free for $\mu$-independent $\rho$ and any $\theta$ smaller than $\pi/d$.
Its image by the map $z(x)$ includes the domain $G(b(\l))$ with $b(\mu)=c\mu^{1+\frac 1{2d}}$ for some positive constant $c$. To see this,
we observe that
$$
\int_{x^*(\mu)}^x\sqrt{A(t)^2-\mu^2}dt=\int_0^{x-x^*(\mu)}\sqrt{A(x^*(\mu)+s)^2-\mu^2}ds
$$
and that, since $A'(x)=-dx^{-d-1}+o(|x|^{-d-1})$ as $x\to \infty$, the Taylor expansion of $A(x^*(\mu)+s)^2-\mu^2$ in $s$ gives 
$$
\sqrt{A(x^*(\mu)+s)^2-\mu^2}\sim \sqrt{2d}\,\mu^{1+\frac 1{2d}}(-s)^{1/2},
$$
as $\mu^{2+\frac 1{d}}s\to 0$. This means that, when $x$ runs from a point $ic_0$ to the right along a line $\im\, x=c_0$ for a small but $\mu$-independent positive $c_0$,
its image $z(x)$ goes from $z(ic_0)$ near $z=0$ with $\re\, z<0$ first to the upper direction and then changes the direction to the right around $z(x^*(\mu))$ 
keeping a distance of order $\mu^{1+\frac 1{2d}}$ from $z(x^*(\mu))$,
and finally goes to infinity above the
horizontal line $\im z=\im z(x^*(\mu))$.
\end{proof}

\begin{example}
\label{exp}
Suppose $A(x)$ satifies (A1), (A2) and 
$$
A(x)={C}e^{-x^\sigma}+r(x) \text{ for } x>1,
$$ 
with $\sigma>0$, $C>0$ and $r(x)=o(e^{-x^\sigma})$, $r'(x)=o(x^{\sigma-1}e^{-x^\sigma})$ as $\re x\to \infty$ in $D_0$.
Then, one can take $a(\l)=c\frac\l {\log\frac 1\l}$ and  $b(\mu)=c\mu(\log\frac 1\mu)^{-1+\frac 1{\sigma}}$ for some positive constant $c$.
\end{example}

\begin{proof}
Here also $C$ is assumed to be 1.

For $\l>0$ small, i.e. $L=\log\frac 1\l$ large, the turning points in $D_0$ are $\pm \left (L+(k+\frac 12)\pi i\right )^{\frac 1{\sigma}}+o(L^{\frac 1\sigma-1})$ 
with some integers $k$ (the distance between two neighboring turning points is of order $L^{\frac 1\sigma-1}$) and  the nearest turning points to the real axis are
$\pm L^{\frac 1{\sigma}}(1\pm\frac{\pi}{2\sigma L}i)+o(L^{\frac 1{\sigma}-1})$.
Hence the domain $D(\rho(\l),\theta(\l))$ has no turning point for $\rho(\l)=\frac {\pi}{4\sigma}L^{\frac 1{\sigma}-1}$ and $\theta (\l)=\frac{\pi }{4\sigma}L^{-1}$. Then we see as in the previous example that
its image by the map $z(x)$ includes the domain $F(a(\l))$ with $a(\l)=c\frac\l L=c\frac\l {\log\frac 1\l}$ for some poisitive constant $c$.

For $\l=i\mu$ with $\mu>0$ small i.e. $M=\log\frac 1\mu$ large,
the turning points in $D_0$ are $\pm \left (M+k\pi i\right )^{\frac 1{\sigma}}+o(M^{\frac 1\sigma-1})$ 
for some integers $k$ and
the nearest turning points  to the real axis (apart from the real ones $\pm x^*(\mu)$) are 
$\pm M^{\frac 1{\sigma}}(1\pm\frac{\pi}{\sigma M}i)+o(M^{\frac 1{\sigma}-1})$.
Hence $D(\rho(\mu), \theta(\mu))\cap \{z\in\C;\im z>0\}$  has no turning point  for $\rho(\mu)=\frac {\pi}{2\sigma}M^{\frac 1{\sigma}-1}$ and  $\theta (\mu)=\frac{\pi }{2\sigma}M^{-1}$.
As in the previous example, we see that its image by the map $z(x)$ includes the domain $G(b(\l))$ with $b(\mu)=c\mu(\log\frac 1\mu)^{-1+\frac 1{\sigma}}$. In fact we have in this case
$$
\sqrt{A(x^*(\mu)+s)^2-\mu^2}\sim\sqrt{2\sigma} \mu M^{\frac 12-\frac 1{2\sigma}}(-s)^{1/2},
$$
and hence
$|\int_{x^*(\mu)}^x\sqrt{A(t)^2-\mu^2}dt|$ is of order $\mu M^{-1+\frac 1{\sigma}}$ when $|x-x^*(\mu)|$ is of order $M^{\frac 1{\sigma}-1}$.
\end{proof}

\begin{corollary}
\label{2.5}
Assume  (A1), (A2) and (A4) with $a(\l)\ge c\l^{\beta}$ for some $\beta>0$ and $c>0$.
Then the reflection coefficient $R(\l,\epsilon)$ is exponentially small with respect to $\epsilon$
uniformly for  $|\l|\ge \epsilon^\alpha$ with any $\alpha<1/\beta$.
In particular, for potentials of Example \ref{pol} and \ref{exp},
$R(\l,\epsilon)$ is exponentially small with respect to $\epsilon$
uniformly for  $|\l|\ge \epsilon^\alpha$ with any $\alpha<1$.
\end{corollary}

Suppose that $\frac\e{b(\mu)}$ is small enough and let $r(\mu,\epsilon):=\log \left (-\frac 1{m(\mu,\e)}\right )$
where the logarithm is defined near 1 with $\log 1=0$.
Then the  Bohr-Sommerfeld quantization condition \eqref{BS} is equivalent to
\begin{equation}
\label{mun}
S(\mu)=(2n+1)\pi \epsilon+i\epsilon r(\mu,\epsilon),
\end{equation}
for some integer $n$. Let then $\mu_n=\mu_n(\epsilon)$ be the (unique) root of \eq{mun}
and let
$\mu_n^{\rm WKB}=\mu_n^{\rm WKB}(\epsilon)$ be the root of the equation
\begin{equation}
\label{mukwkb}
S(\mu)=(2n+1)\pi \epsilon.
\end{equation}

By the previous theorem, we have, as $\epsilon\to 0$ with $\frac{\epsilon^2}{b(\mu_n)}\to 0$,
$$
S(\mu_n)-S(\mu_n^{\rm WKB})=i\epsilon r(\mu_n,\epsilon)=\ord\left(\frac{\epsilon^2}{b(\mu_n)}\right ).
$$
For $\mu z=A(x)$, one has
$$S'(\mu)=2\mu\int_1^{A_0/\mu}\frac {dz}{\sqrt{z^2-1}|A'(x)|}.
$$
In the case of Example \ref{pol}, $A'(x)\sim -d\mu^{1+\frac 1d}z^{1+\frac 1d}$, and  in the case of Example \ref{exp}
$A'(x)\sim -\sigma\mu z(\log\frac 1{\mu z})^{1-\frac 1\sigma}$, and hence
we have, 
$$
|S'(\mu)|\ge 
\left\{
\begin{array}{l}
c\mu^{-\frac 1d}\text{ (Example \ref{pol})}, \\[8pt]
c\left (\log \frac 1\mu\right )^{-1+\frac 1\sigma}\text{ (Example \ref{exp})},
\end{array}
\right.
$$
for some positive constant $c$. Hence we have the following corollary. 

\begin{corollary}
\label{2.7}
Assume  (A1), (A2), (A3) and (A5) with $b(\mu)\ge c\mu^{\beta}$ for some $\beta>0$ and assume 
also $|S'(\mu)|\ge c\mu^{\gamma}$ for some $c>0$. Then
\begin{equation}
\label{muk}
|\mu_n(\epsilon)-\mu_n^{\rm WKB}(\epsilon)|=o(\epsilon)
\end{equation}
uniformly for  $|\mu_n|\ge \epsilon^\alpha$ with any $\alpha<1/(\beta+\gamma)$.
In particular, 
for potentials of Example \ref{pol}, \eq{muk} holds
uniformly for $|\mu|\ge \epsilon^\alpha$ with any $\alpha<\frac d{d+1}$.
For potentials of Example \ref{exp},  \eq{muk} holds
uniformly for $|\mu|\ge \epsilon^\alpha$ with any $\alpha<1$.
\end{corollary}

In section 6, it will be Theorem \ref{ref1} and the above corollary that 
 will be applied to the focusing non-linear Schr\"odinger equation.

\section{Exact WKB method for the Zakharov-Shabat system}
In this section, we briefly review the exact WKB method applied to our operator $L$.
Here we only assume (A1).

The eigenvalue problem of the operator $L$ can be rewritten in the form
\begin{equation}
\label{M}
\frac \e i\frac d{dx}{\bf u}=M(x,\lambda)
{\bf u},
\quad
M(x,\lambda)=
\left(\begin{array}{cc}
-\l & -iA(x)\\
iA(x) & \l
\end{array}\right)
\end{equation} 
where the unknown function ${\bf u}(x,\epsilon)={}^t(u_1(x,\epsilon), u_2(x,\epsilon))$ is a
column vector,  $\epsilon$ is
a small positive parameter, $\l$ is a complex spectral parameter.

The zeros of $\det M(x,\l)=-A(x)^2-\l^2$ are called {\it turning points}.
Let $\Omega$ be a connected subdomain of $D$ free from turning point.
Then the map $x\mapsto z(x;\alpha)$ defined by
\begin{equation}
\label{phase}
z(x;\alpha) = i\int_\alpha^x \sqrt{
A(t)^2+\l^2}\;dt\,,  
\end{equation}
for a fixed point $\alpha$ is conformal from $\Omega$ to $z(\Omega,\alpha)$.

We also define a function 
$$
H(x)=\left (\frac{A(x)-i\l}{A(x)+i\l}\right )^{1/4},
$$
which is holomorphic in $\Omega$ but multivalued around turning points,
and a matrix valued function
$$
Q(x)=\left (
\begin{array}{cc}
1 & 1 \\
-1 & 1
\end{array}
\right )
\left (
\begin{array}{cc}
H(x)^{-1} & H(x)^{-1} \\
iH(x) & -iH(x)
\end{array}
\right ).
$$

Our WKB solutions are  of the form 
\begin{equation}
{\bf u}_\pm(x,\e)=e^{\pm z(x;\alpha)/\e}Q(x)
\left (
\begin{array}{cc}
0 & 1 \\
1 & 0
\end{array}
\right )^{\frac {1\pm 1}2}
{\bf w}^\pm(z(x;\alpha),\e).
\end{equation}
In the usual WKB theory,
the vector valued symbols ${\bf w}^\pm(z,\e)$ are constructed 
as a power series in the parameter $\e$, which is in general divergent.
Here we use the so-called {\it exact} WKB method along the lines of G\'erard-Grigis \cite{gg}
and Fujii\'e-Lasser-N\'ed\'elec \cite{fln}.
This method consists in the resummation of this divergent series in the following way.

We take a point $x_0$ in $\Omega$ and construct ${\bf w}^\pm(z,\e)$ of the form
\begin{equation}
{\bf w}^\pm(z,\e)=\sum_{n=0}^\infty
{\bf w}_n^\pm(z,\e)=\sum_{n=0}^\infty \left (
\begin{array}{c}
w^\pm_{2n} \\
w^\pm_{2n-1}
\end{array}
\right )=:\left (
\begin{array}{c}
w^\pm_{\rm even} \\
w^\pm_{\rm odd}
\end{array}
\right ),
\end{equation}
where the scalar functions $w_n^\pm$ are defined inductively by
\begin{equation}
w_{-1}^\pm\equiv 0,\quad w_0^\pm\equiv 1,
\end{equation}
and for $n\ge 1$,
\begin{equation}
\left\{
\begin{array}{rl}
\label{recurrence}
\displaystyle{
\frac{d}{dz}w_{2n}^\pm}&=\displaystyle{\,\,\,\H(z)\,w_{2n-1}^\pm}, \\[8pt]
\displaystyle{\left (\frac{d}{dz}\pm\frac 2\e\right )w_{2n-1}^\pm}&=\,\,\,
\displaystyle{\H(z)\,w_{2n-2}^\pm.}
\end{array}
\right.
\end{equation}
with initial conditions at $z_0=z(x_0)$
\begin{equation}
\label{initial}
w_n^\pm|_{z=z_0}=0 \quad(n\ge 1).
\end{equation}
Here we defined
$$
\H(z):=\frac{H'_z(z)}{H(z)}=\frac d{dz}\log H(z)=\frac {\l}2\frac {A'_x(x)}{(A(x)^2+\l^2)^{3/2}}.
$$
Notice that $\H(z)$ is holomorphic in $z(\Omega)$, and if $\beta$ is a turning point of order $k$, it behaves, as $z\to \beta$, like
\begin{equation}
\label{pole}
\H(z)=\frac {\mp ik}{(2k+4)(z-\beta)}(1+\ord((z-\beta)^{2/(k+2)})),
\end{equation}
where $\mp$ corresponds to whether $\beta$ is zero of $A-i\l$ or $A+i\l$ ($\l\ne 0$).

The recurrence
equations  
 uniquely determine (at least in a neighborhood of $x_0$) the sequence of scalar
functions $\{w_n^\pm(z,\e;z_0)\}_{n=-1}^\infty$,  and hence the sequence of vector-valued
functions $\{{\bf w}_n^\pm(z,\e;z_0)\}_{n=0}^\infty$. 

The recursive relations  \eqref{recurrence} ,\eq{initial} can be written in the integral form:
\begin{equation}
\label{int}
w_{2n}=J(w_{2n-1}),\quad w_{2n-1}=I_\pm(w_{2n-2}) \quad(n\ge 1),
\end{equation}
with two integral operators
\begin{equation}
\label{J}
J(f)(z):=\int_{\Gamma}\H(\zeta)f(\zeta)d\zeta,
\end{equation}
\begin{equation}
\label{I}
I_\pm(f)(z):=\int_{\Gamma} e^{\pm2(\zeta-z)/\e}\H(\zeta)f(\zeta)d\zeta,
\end{equation}
where $\Gamma=\Gamma(z;z_0)$ is the image by $z=z(x;\alpha)$ of a path $\gamma(x;x_0)$ in $\Omega$ starting from $x_0$ and ending at $x$.

Thus we have constructed formal solutions, which we write
from now on 
${\bf u}_\pm(x,\e;\alpha,x_0)$, or simply ${\bf u}_\pm(x;\alpha,x_0)$
depending on a base point
$\alpha$ for the phase and a base point $x_0$ for the symbol.
This solution has the following important properties:

\begin{theorem}
\label{wkb}
(i) The formal series  are absolutely convergent in a
neighborhood of $x_0$.

(ii) Let $\Omega_\pm$ be the set of $x\in\Omega$ such that there exists a
path $\gamma(x;x_0)$ from $x_0$ to $x$ in $\Omega$ along which $\pm\re \,z(x)$ increases
strictly (we will call such a path {\it progressive}). Then we have for each $N\in\N$
$$
{\bf w}^\pm-\sum_{n=0}^{N-1}{\bf w}^\pm_{n}=\ord(\e^N),
$$
$$
w^\pm_{\rm even}-\sum_{n=0}^{N-1}w_{2n}^\pm=\ord(\e^N),\quad
w^\pm_{\rm odd}-\sum_{n=0}^{N-1}w_{2n+1}^\pm=\ord(\e^N),
$$
as $\e\to 0$, uniformly in any compact subset of $\Omega_\pm$.
In particular,  there we have
$$
w^\pm_{\rm even}=1+\ord (\e),\quad w^\pm_{\rm odd}=\ord (\e).
$$
(iii) The Wronskian of any two exact WKB solutions with different base points of amplitude are given by
\begin{equation}
\label{wronsky+-}
{\mathcal W}({\bf u}^+(x,\e;\alpha, x_0),{\bf u}^-(x,\e;\alpha,
x_1))=4i\,w_{\rm even}^+(z_1;z_0),
\end{equation}
\begin{equation}
{\mathcal W}({\bf u}^+(x,\e;\alpha, x_0),{\bf u}^+(x,\e;\alpha,
x_1))=-4i\,e^{2z_1/\e}w_{\rm odd}^+(z_1;z_0),
\end{equation}
where $z_j=z(x_j;\alpha)$ for $j=0,1$ and ${\mathcal W}({\bf f},{\bf g})$ is by definition the determinant of
the matrix
$({\bf f},{\bf g})$.
\end{theorem}

\begin{proof}
The proof is almost the same as in references 
\cite{gg}, \cite{fln} and \cite{fr}, so we only point out the essence. 

The main point lies in  the ``transport" equation \eqref{recurrence} or equivalently \eqref{int}.
In the usual WKB construction in powers of $\epsilon$, each coefficient is determined as an integral  of the second derivative of the previous coefficient, which makes the sum divergent in general, whereas
in the above construction, $w_n$ is an integral of $w_{n-1}$ itself, which 
makes the sums $\sum w_{2n}$ and $\sum w_{2n+1}$ convergent.
More precisely,
let $K$ be any compact set in $z(\Omega)$. Then one has an estimate
$$
|w_n^\pm(z,z_0)|\le C(AL)^n/n!
$$
with some positive constant $C$ and
$$
L=\text{diam }(K),\quad A=\max_{z\in K}|\H(z)|\cdot\max(1,e^{2L/h}).
$$

As for the asymptotic property (ii), let us define a norm
$$
\Vert f\Vert :=\sup_{\Gamma(z;z_0)}|f|+\epsilon\sup_{\Gamma(z;z_0)}|f'|
$$
for holomorphic functions $f$ on $z(\Omega)$. For $I_+(f)$, we have, by a change of variable
$s=(t-z)/\epsilon$ and the Taylor expansion of $F(z+sh)$ at $z$ in the integral expression,
\begin{equation}
\label{taylor}
I_+(f)=\frac \e 2(1-e^{2(z_0-z)/\e})(\H f)(z)+\e^2\int^0_{(z_0-z)/\e}se^{2s}ds\int_0^1(\H f)'(z+s\e t)dt.
\end{equation}
It follows from the fact $\re \,z<\re \,z_0$ that
$$\sup_{\gamma(z;z_0)}|J\circ I_+(f)|\le C\e\Vert f\Vert.
$$
Moreover, using that $\frac d{dz}(J\circ I_+(f))=\H I_+(f)$, we obtain
\begin{equation}
\label{est}
\Vert J\circ I_+(f)\Vert\le C\e\Vert f\Vert.
\end{equation}
for some positive constant $C$. Hence we conclude that
$$
\Vert w_{2n}^+\Vert=\Vert (J\circ I_+)^n(1)\Vert\le (C\e)^{n},
$$
$$
\Vert w_{2n+1}^+\Vert=\Vert (J\circ I_+)^n(w_1^+)\Vert\le (C\e)^{n}\Vert w_1^+\Vert,
$$
which prove (ii).

It remains to check (iii). We only prove \eqref{wronsky+-}. From the fact that $\det Q=-4i$, we
immediately have
$$
\begin{array}{l}
{\mathcal W}({\bf u}^+(x,\e;\alpha, x_0),{\bf u}^-(x,\e;\alpha,
x_1)) \\
=\det Q\,\,{\mathcal W}
\left (\left (
\begin{array}{cc}
0 & 1 \\
1 & 0
\end{array}
\right ){\bf w}^+(z,z_0),{\bf w}^-(z,z_1)\right) \\
=-4i\,\,(w^+_{\rm odd}(z;z_0)w^-_{\rm odd}(z;z_1)-w^+_{\rm
even}(z;z_0)w^-_{\rm even}(z;z_1)).
\end{array}
$$
This must be independent of $x$ since the matrix $M$ is trace free.
Hence we can replace $x$ in the right
hand side by a particular point, say $x=x_1$. Then taking the previous 
into account, we get the proof for \eqref{wronsky+-}. The proof  for the other formula is similar.
\end{proof}

\begin{remark}
\label{rho}
The constant $C$ in \eq{est} may depend on the energy $\lambda$. In fact, it becomes large when a turning point
approaches the path $\gamma=\gamma(x;x_0)$. 
More precisely, let $\rho(\l)$ be the distance between the path $\gamma$ and the ``nearest" turning point $x^*$
measured after the map $x\mapsto z=z(x)$ (see \eqref{phase}):
$$
\rho(\l)={\rm dist }(z(x^*),z(\gamma)).
$$
Then we have,   with a constant $C'$ independent of $\lambda$,
\begin{equation}
\label{C}
C=\frac{C'}{\rho(\l)}.
\end{equation}

  This fact has already been proved and used in  the Schr\"odinger case in \cite{gg} and \cite{fr}
  for the study of eigenvalues or resonances close to a barrier top of the potential, where the two turning points
  near the non-degenerate maximal point ``pinch" a path along which the wronskian of two solutions on the opposite side of the barrier
  should be computed.
  Here we use this fact for the study of eigenvalues and the reflection coefficient near $\l=0$.
  We briefly sketch the proof of \eqref{C} below.
  
  To estimate $J\circ I_+(f)$, we should study integrals of type
  $$
  I_1= \int_{z_1}^z |\H(\tau)||\H(\tau+ths)|d\tau,\quad   I_2= \int_{z_1}^z |\H(\tau)||\H'(\tau+ths)|d\tau.
  $$
  Because of the Cauchy-Schwarz inequality and \eq{taylor} we only need to estimate
  $
  \int_{z_1}^z |\H(t)|^2dt,\quad \int_{z_1}^z|\H'(t)|^2dt.
  $
  Since the singularity of the function $\H(z)$ at the image of turning points is like \eq{pole},
  these two integrals are of type
  $ \int_{\R}\frac{dt}{t^2+\rho^2}$, $ \int_{\R}\frac{dt}{(t^2+\rho^2)^2}$
  respectively. This gives 
  $I_1\le\frac c\rho$, $ I_2\le \frac c{\rho^2}$,
  and consequently
  $$
  \sup |J\circ I_+(f)|\le C_1\left (\frac\e\rho+(\frac \e\rho)^2\right )\sup |f|+C_1\frac{\e^2}\rho\sup|f'|.
  $$
  Using again $\frac d{dz}(J\circ I_+(f))=\H I_+(f)$, we get \eq{C} in the estimate \eq{est}.
\end{remark}

The level curves of $\re \,z(x)$ in the $x$-complex plane
$$
\{x\in D; \re z(x)=\text{const.}\}
$$
are called {\it Stokes curves}.  
(Sometimes it is the level curves of $\im z(x)$, especially only those passing through turning points, 
that are called Stokes curves,
but we  employ the former definition.)
The geometric configuration of Stokes curves is useful for us to know the domain of validity of the
asymptotic expansion of WKB solutions.

Notice that from a simple turning point exactly three Stokes curves emanate and the angles between 
two of them are all $2\pi/3$ at that point.

\section{Reflection coefficient}
Here we compute the reflection coefficient for real positive $\l$.
The computation for negative $\l$ is similar.

Under the assumption (A1), there exist a pair of solutions ${\bf f}^r_+(x,\e)$, ${\bf f}^r_-(x,\e)$ which behave, as $\re \,x\to \infty$ in $D$, like
$$
{\bf f}^r_+\sim \left (
\begin{array}{c}
0 \\
e^{i\l x/\e}
\end{array}
\right ), \quad
{\bf f}^r_-\sim\left (
\begin{array}{c}
e^{-i\l x/\e} \\
0
\end{array}
\right ),
$$
as well as a pair of solutions ${\bf f}^l_+(x,\e)$, ${\bf f}^l_-(x,\e)$ which behave, as $\re x\to -\infty$ in $D$, like
$$
{\bf f}^l_+\sim \left (
\begin{array}{c}
0 \\
e^{i\l x/\e}
\end{array}
\right ), \quad
{\bf f}^l_-\sim\left (
\begin{array}{c}
e^{-i\l x/\e} \\
0
\end{array}
\right ).
$$
These solutions are called  {\it Jost solutions}.
Each of these pairs is uniquely determined and makes a basis of solutions.
Let $T(\l, \e)$ be the 2$\times 2$ constant matrix depending on $\l$ and $\e$ expressing the change of basis
of these two pairs:
\begin{equation}
({\bf f}^l_+,{\bf f}^l_-)=({\bf f}^r_+,{\bf f}^r_-)T.
\end{equation}
Then $T$ is of the form
\begin{equation}
T(\lambda,\e)=
\left (
\begin{array}{cc}
a(\l,\e) & b^*(\l,\e) \\
b(\l,\e) & a^*(\l,\e)
\end{array}
\right )
\end{equation}
where $a^*$, $b^*$ denote the complex conjugates of $a,b$. The reflection coefficient $R(\l,\e)$ is by definition
\begin{equation}
R(\l,\e)=\frac{b(\l,\e)}{a(\l,\e)}.
\end{equation}
It is easy to see  that it can be expressed by
wronskians of Jost solutions:
\begin{equation}
R(\l,\e)=\frac{\W({\bf f}^r_+,{\bf f}^l_+)}{\W({\bf f}^l_+,{\bf f}^r_-)}.
\end{equation}

We construct the Jost solutions as exact WKB solutions.
We define four exact WKB solutions:
\begin{equation}
\widetilde{\bf v}^r_\pm(x,\e)=e^{\pm z^r(x)/\e}Q(x)
\left (
\begin{array}{cc}
0 & 1 \\
1 & 0
\end{array}
\right )^{\frac {1\pm 1}2}
{\bf w}_r^\pm(x,\e),
\end{equation}
\begin{equation}
\widetilde{\bf v}^l_\pm(x,\e)=e^{\pm z^l(x)/\e}Q(x)
\left (
\begin{array}{cc}
0 & 1 \\
1 & 0
\end{array}
\right )^{\frac {1\pm 1}2}
{\bf w}_l^\pm(x,\e),
\end{equation}
where
the phase function are 
$$
z^r(x):=i\l x+i\int_{+\infty}^x(\sqrt{A(t)^2+\l^2}-\l)dt,
$$
$$
z^l(x):=i\l x+i\int_{-\infty}^x(\sqrt{A(t)^2+\l^2}-\l)dt,
$$
which are both primitives of $i\sqrt{A(x)^2+\l^2}$.

As base point of the symbol ${\bf w}_r^\pm(x,\e)$ and ${\bf w}_l^\pm(x,\e)$, we choose $e^{\pm i\theta_0}\infty$
and $e^{i(\pi \mp \theta_0)}\infty$ respectively.
We recall that $\theta_0$ is the positive angle of the sector at infinitiy of the domain $D$ (see the assumption (A1)).
More precisely, we take, as the contour for the integral operators $I_\pm, J$ the image by the map $x\mapsto z^r(x)$  (resp.  $x\mapsto z^l(x)$)
of a curve
from $e^{\pm i\theta_0}\infty$ (rest. $e^{i(\pi \mp \theta_0)}\infty$) to $x$, which is transverse to the Stokes curves.
This is possible for any $x\in D(\mu, R, \theta_0)$ if $\delta$ is sufficiently small and 
$R$ is sufficiently large, because $D(\mu, R, \theta_0)$ contains no turning point, the Stokes curves are asymptotic
to horizontal lines as $\re x\to \pm\infty$
and the real axis is itself a Stokes curve. We take a branch for the functions $(A(x)\pm i\l)^{1/2}$ and $(A(x)\pm i\l)^{1/4}$ in 
such a way that
the argument of these functions tends to 0 as $\l\to 0$
(recall that $A(x)$ is positive).
Then the real part of the phase $z^r(x), z^l(x)$ or $z(x;\alpha)$ for any $\alpha$ increases as $\im \,x$ decreases.
Remark also that, by this determination, one has
$$
H(x)\to e^{-\frac\pi 4 i\,{\rm sgn}\l}\quad {\rm as}\,\,\, |x|\to\infty.
$$
Hence, for $\l>0$, we have
$$
Q(x)\to 2e^{\frac \pi 4 i}
\left (
\begin{array}{cc}
1 & 0 \\
0 & -1 
\end{array}
\right )\quad {\rm as}\,\,\, |x|\to\infty.
$$

These exact WKB solutions have the following trivial relations with the Jost solutions:
\begin{equation}
{\bf f}^r_\pm=\mp 2e^{\pi i/4}\widetilde{\bf v}^r_\pm\quad
{\bf f}^l_\pm=\mp 2e^{\pi i/4}\widetilde{\bf v}^l_\pm.
\end{equation}

We further modify slightly our exact WKB solutions. Let ${\bf v}^r_\pm$, ${\bf v}^l_\pm$ be the exact WKB solutions defined just like $\widetilde{\bf v}^r_\pm$, $\widetilde{\bf v}^l_\pm$ but
with $z(x;0)$ for the phase. Then we obviously have
\begin{equation}
\widetilde{\bf v}^r_\pm=e^{\pm z_r(0)/\e}{\bf v}^r_\pm,\quad
\widetilde{\bf v}^l_\pm=e^{\pm z_l(0)/\e}{\bf v}^l_\pm.
\end{equation}

In terms of these WKB solutions, the reflection coefficient is expressed by
\begin{equation}
\label{refwronskian}
R(\l,\e)=-\frac{\W({\bf v}^r_+,{\bf v}^l_+)}{\W({\bf v}^l_+,{\bf v}^r_-)}e^{2z^r(0)/\e},
\end{equation}
where we recall that
\begin{equation}
\label{zr0}
z^r(0)=-i\int^{+\infty}_0(\sqrt{A(t)^2+\l^2}-\l)dt.
\end{equation}

The wronskians appearing in \eq{refwronskian} can be expressed by the functions $w_{even}^\pm$ and
$w_{odd}^\pm$ using Theorem \ref{wkb}.

First, $\W({\bf v}^l_+,{\bf v}^r_-)$ is given by
$$
\W({\bf v}^l_+,{\bf v}^r_-)=4iw_{even}^+(e^{-i\theta_0}\infty; e^{i(\pi - \theta_0)}\infty).
$$

On the other hand, we should express the wronskian $\W({\bf v}^r_+,{\bf v}^l_+)$ via another basis of solutions since there is no progressive path between $e^{i\theta_0}\infty$ and $e^{i(\pi - \theta_0)}\infty$.
We take exact WKB solutions defined with the phase \eqref{phase}
$$
{\bf v}^0_+(x,\e):={\bf u}_+(x,\e;;0,x_+),
$$ 
$$
{\bf v}^0_-(x,\e):={\bf u}_-(x,\e;;0,x_-),
$$ 
where we take $x_\pm\in D\cap\C_\pm$ near the origin.
We can write
$$
\left \{
\begin{array}{l}
{\bf v}^r_+=c^r_+{\bf v}^0_++c^r_-{\bf v}^0_-, \\[8pt]
{\bf v}^l_+=c^l_+{\bf v}^0_++c^l_-{\bf v}^0_-,
\end{array}
\right.
$$
with
$$
c^r_+=\frac{\W({\bf v}^r_+,{\bf v}^0_-)}{\W({\bf v}^0_+,{\bf v}^0_-)},\quad
c^r_-=\frac{\W({\bf v}^0_+,{\bf v}^r_+)}{\W({\bf v}^0_+,{\bf v}^0_-)},
$$
$$
c^l_+=\frac{\W({\bf v}^l_+,{\bf v}^0_-)}{\W({\bf v}^0_+,{\bf v}^0_-)},\quad
c^l_-=\frac{\W({\bf v}^0_+,{\bf v}^l_+)}{\W({\bf v}^0_+,{\bf v}^0_-)}.
$$
Since
$$
\W({\bf v}^r_+,{\bf v}^l_+)=(c^r_+c^l_--c^r_-c^l_+)\W({\bf v}^0_+,{\bf v}^0_-),
$$
we have
$$
\W({\bf v}^r_+,{\bf v}^l_+)=\frac{\W({\bf v}^r_+,{\bf v}^0_-)\W({\bf v}^0_+,{\bf v}_+^l)
-\W({\bf v}_+^l,{\bf v}^0_-)\W({\bf v}^0_+,{\bf v}_+^r)}{\W({\bf v}^0_+,{\bf v}^0_-)}.
$$
The wronskian formulae of Theorem \ref{wkb} give us the following expressions.
$$
\begin{array}{l}
\W({\bf v}^0_+,{\bf v}^0_-)=4iw_{even}^+(x_-;x_+), \\[8pt]
\W({\bf v}^r_+,{\bf v}^0_-)=4iw_{even}^+(x_-; e^{i\theta_0}\infty), \\[8pt]
\W({\bf v}^l_+,{\bf v}^0_-)=4iw_{even}^+(x_-;e^{i(\pi - \theta_0)}\infty), \\[8pt]
\W({\bf v}^0_+,{\bf v}^r_+)=4iw_{odd}^+(x_+;e^{i\theta_0}\infty)e^{2\sigma/\e}, \\[8pt]
\W({\bf v}^0_+,{\bf v}^l_+)=4iw_{odd}^+(x_+;e^{i(\pi - \theta_0)}\infty)e^{2\sigma/\e},
\end{array}
$$
where
\begin{equation}
\label{s}
\sigma=z(x_+;0)=i\int_0^{x_+}\sqrt{A(t)^2+\l^2}dt.
\end{equation}
Notice that $\sigma$ has a negative real part since $\im\, x_+>0$.

Summing up, we arrive at the following formula for the reflection coefficient.
\begin{proposition}
\label{refwronsky}
Let $\lambda>0$ and $\sigma$ defined by \eq{s}.
Then  one has
$$
R(\lambda,\e)=-e^{2(\sigma+z_r(0))/\e}\times
\hspace{8cm}
$$
$$
\frac{w_{even}^+(x_-; e^{i\theta_0}\infty)w_{odd}^+(x_+;e^{i(\pi - \theta_0)}\infty)
-w_{even}^+(x_-;e^{i(\pi - \theta_0)}\infty)w_{odd}^+(x_+;e^{i\theta_0}\infty)}
{w_{even}^+(x_-;x_+)w_{even}^+(e^{-i\theta_0}\infty; e^{i(\pi - \theta_0)}\infty)}
$$
\end{proposition}

\subsection{Proof of Theorem \ref{ref1}}
In this theorem, it is assumed that $\lambda>\delta$ for some positive $\delta$ independent of $\e$.

First recall that $z_r(0)$ is purely imaginary (see \eq{zr0}) and does not affect the absolute value of the reflection coefficient.

Next, as long as $\im \,x_+$ is positive and small enough, we have $\re\,\sigma<0$ and
we can find a progressive path for each couple of points in $w^+_{even}$ and $w^+_{odd}$
of the formula of Proposition \ref{refwronsky}, which means that these quantities
are $1+\ord (\e)$ and $\ord(\e)$ respectively as $\e\to 0$.

This gives the proof of Theorem \ref{ref1}.

\subsection{Proof of Theorem \ref{ref2}}

The assumption (A4) together with the obvious fact that  $z(x)$ maps a $\l$-independent neighborhood of $x=0$ to a $\l$-independent neighborhood of $z=0$ 
imply that the image of $D(\rho(\l),\theta(\l))$ includes a domain of the form
$$
\left\{z\in\C; |\im\, z|>\frac {2|\re\,z|}{a(\l)}-c\right\}
$$
for some positive constant $c$. 

In the computation of the asymptotic behavior of the wronskians appearing in Proposition \ref{refwronsky}, for example
$w_{even}^+(x_-; e^{i\theta_0}\infty)$ and $w_{odd}^+(x_+;e^{i\theta_0}\infty)$, 
we take, as the integral contour for \eq{J} and \eq{I},  the half lines $\{\im\, z=-\frac {2\re\,z}{a(\l)}+\frac c2, \re\,z\le \frac {ca(\l)}4\}$ and $\{\im\, z=-\frac {2\re\,z}{a(\l)}-\frac c2, \re\,z\le-\frac {ca(\l)}4\}$ respectively, oriented in such a way that $\re\, z$ increases (we take $x_\pm$ so that $z(x_\pm)=\mp \frac {ca(\l)}4$). Then the quantity $\rho(\l)$ in Remark \ref{rho}, which measures the distance from the contour to the nearest turning points, is estimated from below by a constant multiple
of $a(\l)$. This proves Theorem \ref{ref2}.




\section{Eigenvalues}

In this section, we study the eigenvalue problem of the operator $L$.
It is known (\cite{ks}, \cite{k}) that for our kind of potential $A(x)$ 
the eigenvalues are all simple and purely imaginary with imaginary part in $[-A(0),A(0)]$.

For  $\l=i\mu$, $\mu\in (0,A(0))$, there are exactly two simple turning points $x^*(\mu)>0$ and $-x^*(\mu)$
on the real axis.  
There is no other turning point in the complex domain $D(\mu_0,R_0,\theta_0)$ if we take $\mu_0$ sufficiently small, $R_0$ sufficiently large and $\theta_0$ sufficiently small depending on each $\mu>0$.

The interval $[-x^*(\mu),x^*(\mu)]$ is a Stokes curve on which $A(x)^2-\mu^2\ge 0$.
There are two other Stokes curves emanating from each of these turning points.

We take two branch cuts along  Stokes curves, one from $x^*(\mu)$ with angle $\pi/3$ and 
the other from $-x^*(\mu)$ with angle $4\pi/3$, and determine the branch of
$(A(x)^2-\mu^2)^{1/2}$ and $(A(x)\pm\mu)^{1/4}$ so that they are all real and positive on the interval $[-x^*(\mu),x^*(\mu)]$. Then automatically 
$$
\begin{array}{l}
(A(x)^2-\mu^2)^{1/2}\in i\R_+\text{ in }(-\infty, -x^*(\mu)]\cup [x^*(\mu),\infty),\\[8pt]
(A(x)-\mu)^{1/4}\in e^{i\pi/4}\R_+\text{ in }(-\infty, -x^*(\mu)]\cup [x^*(\mu),\infty),\\[8pt]
(A(x)+\mu)^{1/4}\in \R_+\text{ on } \R.
\end{array}
$$

Now we define several exact WKB solutions. 
The 5 Stokes curves in $D$ divide the domain $D$ into 4 connected regions $D_r$, $D_l$,
$D_u$ and $D_d$ (if $D$ is chosen sufficiently small as mentioned above). The regions $D_r$ and $D_l$ include $(x^*(\mu),+\infty)$ and $(-\infty, -x^*(\mu))$
respectively, and $D_u\subset \C_+=\{x\in \C;\im \,x>0\}$ and $D_d\in \C_-$ share $(-x^*(\mu),x^*(\mu))$ as a part of their boundary.
We take four base points $x_1\in D_l$, $x_2\in D_u$, $x_3\in D_d$, $x_4\in D_r$.
With the above determination, the real part of $z(x)$ increases along curves from $x_2$ to $x_1$, from $x_2$ to $x_3$, and from $x_4$ to $x_3$. 
Taking this into account, we define six exact WKB solutions.
$$      
{\bf v}_{1}(x,\epsilon,\mu):={\bf u}^{-}(x,\epsilon;x^*(\mu),x_{1}),
$$
$$
{\bf v}_{2}(x,\epsilon,\mu):={\bf u}^{+}(x,\epsilon;x^*(\mu),x_{2}),\quad {\widetilde {\bf v}_{2}}(x,\epsilon,\mu):={\bf u}^{+}(x,\epsilon;-x^*(\mu),x_{2}),
$$
$$
{\bf v}_{3}(x,\epsilon,\mu):={\bf u}^{-}(x,\epsilon;x^*(\mu),x_{3}),\quad {\widetilde{\bf v}_{3}}(x,\epsilon,\mu):={\bf u}^{-}(x,\epsilon;-x^*(\mu),x_{3}),
$$
$$
{\bf v}_{4}(x,\epsilon,\mu):={\bf u}^{+}(x,\epsilon;-x^*(\mu),x_{4}).
$$

Exactly as in \cite{gg} in the Schr\"odinger case, we know
\begin{lemma}
For each $\epsilon>0$, ${\bf v}_{1}(x,\epsilon,\mu)\in (L^2(\R_-))^2$, ${\bf v}_{4}(x,\epsilon,\mu)\in (L^2(\R_+))^2$.
\end{lemma}
\begin{remark}
We could have chosen the base point $-\infty$ for ${\bf v}_1$, $\tilde {\bf v}_1$ instead of $x_1$,
and $+\infty$ for
${\bf v}_4$, $\tilde {\bf v}_4$ instead of $x_4$, as in the study of the reflection coefficient.
\label{infty}
\end{remark}
This lemma immediately implies 
\begin{proposition}
 $\lambda=i\mu$ is an eigenvalue if and only if the wronskian between ${\bf v}_{1}(x,\epsilon,\mu)$ and ${\bf v}_{4}(x,\epsilon,\mu)$ vanishes.
\end{proposition}

The pair of solutions $({\bf v}_{2},{\bf v}_{3})$ and $(\widetilde{\bf v}_{2},\widetilde{\bf v}_{3})$ form bases, since we have from Theorem \ref{wkb} (iii) \eqref{wronsky+-}
\begin{equation}
\label{23}
\mathcal{W}({\bf v}_{2},{\bf v}_{3})= \mathcal{W}(\widetilde{\bf v}_{2},\widetilde{\bf v}_{3})
=4iw_{even}^{+}(x_{3},\epsilon;x_{2})=4i+\ord (\epsilon)\ne 0.
\end{equation}
These two bases have the following trivial relations.
\begin{lemma}
\label{c2c3}
Let $S(\mu)$ be the action integral defined by \eqref{action}.
Then one has
\begin{align*}
{\widetilde {\bf v}_{2}} = e^{-iS(\mu)/\epsilon}{\bf v}_{2},~~~{\widetilde {\bf v}_{3}} = e^{iS(\mu)/\epsilon}{\bf v}_{3}.
\end{align*}
\end{lemma}

In order to compute the asymptotic behavior of $\mathcal W({\bf v}_{1},{\bf v}_{4})$, we need to express ${\bf v}_{1}$ and ${\bf v}_{4}$ in terms of ${\bf v}_{2}$ and ${\bf v}_{3}$ or
 in terms of $\widetilde{\bf v}_{2}$ and $\widetilde{\bf v}_{3}$:
\begin{align*}
{\bf v}_{1} = c_{2}(\epsilon,\mu){\bf v}_{2}+c_{3}(\epsilon,\mu){\bf v}_{3} ,~~~{\bf v}_{4} = {\widetilde c}_{2}(\epsilon,\mu){\widetilde {\bf v}_{2}}+{\widetilde c}_{3}(\epsilon,\mu){\widetilde {\bf v}_{3}}.
\end{align*}
 The coefficients are written in terms of  wronskians:
\begin{align*}
c_{2} = \frac{\mathcal{W}({\bf v}_{1},{\bf v}_{3})}{\mathcal{W}({\bf v}_{2},{\bf v}_{3})},~~~c_{3} = \frac{\mathcal{W}({\bf v}_{1},{\bf v}_{2})}{\mathcal{W}({\bf v}_{3},{\bf v}_{2})},\\
{\widetilde c}_{2} = \frac{\mathcal{W}({\bf v}_{4},{\widetilde {\bf v}_{3}})}{\mathcal{W}({\widetilde {\bf v}_{2}},{\widetilde {\bf v}_{3}})},~~~{\widetilde c}_{3} = \frac{\mathcal{W}({\bf v}_{4},{\widetilde {\bf v}_{2}})}{\mathcal{W}({\widetilde {\bf v}_{3}},{\widetilde {\bf v}_{2}})}.
\end{align*}
Thus a computation of $\mathcal W({\bf v}_{1},{\bf v}_{4})$ leads us to the following quantization condition of eigenvalues in terms of wronskians of exact WKB solutions:
\begin{proposition}\label{qc}
Let $m(\mu,\e)$ be a function defined by
$$
m(\mu,\e)=\frac{\mathcal{W}({\bf v}_{1},{\bf v}_{3})\mathcal{W}({\bf v}_{4},{\widetilde {\bf v}_{2}})}
{\mathcal{W}({\bf v}_{1},{\bf v}_{2})\mathcal{W}({\bf v}_{4},{\widetilde {\bf v}_{3}})}.
$$
Then
$\l=i\mu$ is an eigenvalue if and only if 
$$
m(\mu,\e)e^{2iS(\mu)/\e}=1.
$$
\end{proposition}

Now we study the asymptotic behavior of the wronskians appearing in Proposition \ref{qc}
as $\e\to 0$ using Theorem \ref{wkb}. 
\begin{lemma}
It holds that
$$
\mathcal{W}({\bf v}_{1},{\bf v}_{2})=-4i w_{even}^{+}(x_{1},\epsilon;x_{2}),
\quad
\mathcal{W}({\bf v}_{4},{\widetilde {\bf v}_{3}})
=4iw_{even}^{+}(x_{3},\epsilon;x_{4}),
$$
$$
\mathcal{W}({\bf v}_{1},{\bf v}_{3})= 4w_{even}^{+}(x_{3},\epsilon;\hat x_{1}),
\quad
\mathcal{W}({\bf v}_{4},{\widetilde {\bf v}_{2}})= -4w_{even}^{+}(\hat x_{4},\epsilon;x_{2}),
$$
where $\hat x_1$ is the same point as $x_1$ but on the Riemann sheet continued
from $x_3$ crossing the branch cut from $-x^*(\mu)$ and similarly $\hat x_4$
is the same point as $x_4$ but on the Riemann sheet continued from $x_2$ crossing the branch cut from $x^*(\mu)$. Consequently, we have
\begin{equation}
\label{mmh}
m(\mu,h)=-\frac{w_{even}^{+}(x_{3},\epsilon;\hat x_{1})w_{even}^{+}(\hat x_{4},\epsilon;x_{2})}{w_{even}^{+}(x_{1},\epsilon;x_{2})w_{even}^{+}(x_{3},\epsilon;x_{4})}.
\end{equation}
\end{lemma}
\begin{proof}
The formulas for $\mathcal{W}({\bf v}_{1},{\bf v}_{2})$ and $\mathcal{W}({\bf v}_{4},{\widetilde {\bf v}_{3}})$ follow difrectly from \eqref{wronsky+-}.

For the computation of $\mathcal{W}({\bf v}_{1},{\bf v}_{3})$, we should be careful with the branch cut lying between $x_1$ and $x_3$.
In order to compute the wronskian on a curve along which $\re z$ is strictly increasing, we have to rewrite ${\bf v}_1$, say, on the Riemann sheet continued from $x_3$
along this curve crossing the cut.

Let $x$ be a point near $x_1$ and $\hat x$ the same point as $x$ but on the Riemann sheet mentioned above.
Then we have, writing $\alpha=x^*(\mu)$ for simplicity,
\begin{align*}
z(x;\alpha) = - z(\hat{x};\alpha), ~~H(x) = iH(\hat{x}), ~~{\bf w}^{\pm}(x,\epsilon;x_1) = {\bf w}^{\mp}(\hat{x},\epsilon;x_1) .
\end{align*}
In fact, the first identity is obvious. For the second one, remark that the turning point $\beta:=-x^*(\mu)$ is a zero of $A(x)-\mu$.
The third one can be seen from the first one and \eqref{recurrence}. It follows that
\begin{align*}
{\bf v}_{1} &= {\bf u}^{-}(x,\epsilon;\alpha,x_{1}) = e^{- z(x;\alpha)/\epsilon}Q(x)
{\bf w}^{-}(x,\epsilon;x_{1}) \\
&=  e^{+ z(\hat{x};\alpha)/\epsilon}
\begin{pmatrix}
-i & 0 \\
0 & i
\end{pmatrix}
Q(\hat{x})
{\bf w}^{+}(\hat{x},\epsilon;x_{1})\\
&= -ie^{+ z(\hat{x};\alpha)/\epsilon}
\begin{pmatrix}
1 & 0 \\
0 & -1
\end{pmatrix}
Q(\hat{x})
{\bf w}^{+}(\hat{x},\epsilon;x_{1}),
\end{align*}
Since
$\begin{pmatrix}
1 & 0 \\
0 & -1
\end{pmatrix}Q(\hat{x})
=Q(\hat{x})\begin{pmatrix}
0 & 1 \\
1 & 0
\end{pmatrix}$,
we finally get
\begin{align*}
{\bf v}_{1} = -i{\bf u}^{+}(\hat{x},\epsilon;\alpha,\hat x_{1}) .
\end{align*}
Now we can compute the wronskian $\mathcal{W}({\bf v}_{1},{\bf v}_{3})$ by the formula \eqref{wronsky+-} and 
$$
\mathcal{W}({\bf v}_{1},{\bf v}_{3}) = \mathcal{W}\left(-i{\bf u}^{+}(\hat{x},\epsilon;\alpha,\hat x_{1}), {\bf u}^{-}(\hat{x},\epsilon;\alpha,x_{3})\right) 
= 4w_{even}^{+}(x_{3},\epsilon;\hat x_{1}).
$$

In the same way, we rewrite $ {\bf v}_{4}$ using the branch continued from $x_{2}$ to $x_{4}$ crossing the cut emanating from $x^*(\mu)$;
$$
{\bf v}_{4}=-i{\bf u}^{-}(\hat{x},\epsilon;\beta,\hat x_{4}).
$$
This enables to compute the wronskian 
$$
\mathcal{W}({\bf v}_{4},{\widetilde {\bf v}_{2}})= -4w_{even}^{+}(\hat x_{4},\epsilon;x_{2}).$$
\end{proof}

\subsection{Proof of Theorem \ref{ev1}}
For the time let us ignore  the assumption (A3) and simply assume that $\mu$ stays away from $A(0)$;
$\mu\in[\delta, A(0)-\delta]$ for  some $\e$-independent positive constant $\delta$.
Then the configuration of Stokes curves in $D$ is as explained at the beginning of 
this section if $\mu_0$, $R_0$, $\theta_0$ are suitably chosen.

Notice that the asymptotic behavior of the quantities $w_{even}^{+}$ in the right hand side of \eq{mmh} are all well known by Theorem \ref{wkb} (ii) to be
$1+\ord (\epsilon)$ 
uniformly with respect to $\mu$ in the interval $[\delta, A(0)-\delta]$. In fact, 
there exist  progressive paths from $x_2$ to $x_1$, from $x_4$ to $x_3$,
 from $\hat x_1$ to $x_3$ and from $x_2$ to $\hat x_4$. Hence we obtain
\begin{equation}
\label{mmh-1}
m(\mu,\epsilon)=-1+\ord(\e),
\end{equation}

Next we consider the case where $\mu\to A(0)$.
In this limit,  the two turning points $x^*(\mu)$ and $-x^*(\mu)$ coalesce at the origin $x= 0$ and they become a double turning point. 

Under the additional condition (A3), however, there is no other complex turning point 
converging to the origin and the Stokes geometry 
does not change 
in $D(\mu_0,R_0,\theta_0)$ with small enough $\e$-independent $\mu_0$
  except that the two turning points get closer and closer.

Moreover, it is important to notice that the four paths from $x_2$ to $x_1$, from $x_4$ to $x_3$,
 from $\hat x_1$ to $x_3$ and from $x_2$ to $\hat x_4$
 are not {\it pinched} by these two turning points. This fact implies that
the asymptotic formula \eq{mmh-1} holds true also for such energies $\mu$.

\subsection{Proof of Theorem \ref{ev2}}
Here we take $x_1=-\infty$ and $x_4=+\infty$ (see Remark \ref{infty}). 

For the computation of the asymptotic behavior of the wronskians appearing in \eq{mmh}, say $w_{even}^{+}(\hat x_{4},\epsilon;x_{2})$,
we take, as the integral contour for \eq{J} and \eq{I}, 
a curve  from $z(x_2)$ ($x_2$ should be taken so that $z(x_2)\in G(b(\mu))$) passing inside the tube $G(b(\mu))$ with increasing real part
 to $z(\hat x_4)$ (i.e. going to $\infty$ passing through the tube in the upper half plane of $G(b(\mu))$).

Then the quantity $\rho(\l)$ in Remark \ref{rho}, which measures the distance from the contour to the nearest turning points, is estimated from below by a constant multiple
of $b(\l)$. This proves Theorem \ref{ev2}.



\bigskip

\section{Application of the exact WKB results to the focusing nonlinear Schr\"odinger equation}

We are now ready to give a rigorous justification of the asymptotics \eqref{asymptotics} stated in the introduction 
for a general class of initial data, $without$ having to replace them by their WKB approximation. 
The proof is a variant of Chapter 4 in \cite{kmm} (see also the first appendix below 
for the definition of the Riemann-Hilbert problem
and the second appendix for the deformations implemented involving a well-chosen phase function $g$),
as improved in \cite{lm}.
Having estimated the error of the WKB approximation at the level of the scattering data,
this error can be built into the Riemann-Hilbert analysis as another layer of approximtion.
\footnote{This is the strategy we explicitly proposed in \cite{kmm}.}

In \cite{kmm} we considered separately two complementary sets for $\lambda$:  a disc centered at $0$ with radius of order
$\epsilon^{\delta}$,  for some $\delta \in (1/2,1)$ and the complement of that disc.
In the complement of the disc we needed some delicate rigorous estimates  
while inside the disc we only relied on some symmetry properties.  
This was enough to be able to approximate the "given" Riemann-Hilbert problem (for the WKB-approximated pure soliton data)
by an approximate Riemann-Hilbert problem
which we could eventually analyse. The approximation was good enough away from zero (and this is all we are interested 
in because the solution of the NLS equation only depends 
on the behaviour of the solution of the Riemann-Hilbert problem near infinity).

Now, it turns out that the restriction that  $\delta \in (1/2,1)$ is too 
strong if we want to use the results of the previous sections.
Conveniently, there is an improvement of our argument in \cite{lm}, based on the observation that the approximation  of the
Blaschke product (see below) involving the (WKB approximations of the)  
eigenvalues by a logarithmic integral 
is better done separately, with slightly different choices,  
in different sides of the segment $[-iA(0), iA(0)]$. 
\footnote{As we explain in \cite{k2} and \cite{k3} this corresponds to different sheets of the logarithmic
kernel in this integral. Different approximations 
are required in different sheets for best results.}

Consequently, the choice of the small circle separating the two cases for $\lambda$ above can 
be allowed to have a radius independent of
$\epsilon$ as long as it is small enough in the sense of meeting some requirement spelled out in \cite{lm} 
(which in turn is imposed
by the asymptotic analysis of the approximate Riemann-Hilbert problem).
\footnote{We could still ignore the improvement in \cite{lm} and give a different argument involving
different circles in different steps of the Riemann-Hilbert sequence. We feel that the argument would become 
a bit more cumbersome.}§
This is certainly good enough for our purposes here.

This point in \cite{lm} is only explicitly detailed  
for the very specific case 
where the distribution of the eigenvalues is uniform
(the Satsuma-Yajima case) and where there is no reflection coefficient,
but it is clear that after an obvious modification it can apply to our general case. 
\footnote{The function $\theta^0(\lambda) = i \pi \lambda +\pi A$ of \cite{lm} has to be replaced
by the integral of the 
eigenvalue density in the general case.}

Under the assumptions (A1), (A2), (A3), (A5),  with $b(\mu)\ge c\mu^{\beta}$ for some $\beta>0$ and  assuming
that the action integral $S$ satisfies
$|S'(\mu)|\ge c\mu^{\gamma}$ for some $c>0$, with $\beta + \gamma <2$, we then have the following.

\begin{proposition}
Assuming the existence of the finite gap ansatz 
and also that the density of eigenvalues admits an analytic extension in the
upper half-plane, 
the asymptotics \eqref{asymptotics} stated in the introduction are valid.
\end{proposition}

\begin{proof}
The proof is a variant of Chapter 4 in \cite{kmm}, as improved in \cite{lm}.
There are two modifications. 

First, because of the non-triviality of the reflection coefficient, the Riemann-Hilbert
contour is augmented by the real line (oriented from left to right so that the ``+"-side is on top). 
Still, we use exactly the same $g$ as in \cite{kmm}.
The important condition  
$$
\aligned
g(\lambda^*)+g(\lambda)^* = 0 \text{ for
all } \lambda\in \Bbb C\setminus (C\cup C^*),
\endaligned
$$
(recall that $C$ is a contour lying in the closed upper half lane, including $0$ and encricling the segment $(0,iA_0]$
and approaching $0$ at an angle strictly between $0$ and $\pi/2$ from the first quadrant) 
implies that $g(\lambda)$ is imaginary on the real line. Here  * denotes complex conjugation.

As a consequence, the fact that the reflection coefficient is exponentially small outside of a small 
open disc $D$ with center $0$ and radius small enough for the asymptotic analysis to go through
means that the jump matrix on 
$\Bbb R \setminus D$ is of the form $I+\text{ \it exponentially small}$ uniformly.

Within the disc the important observation is that the (no more trivial) jump matrix still respects the Schwarz reflection 
symmetry conditions and the positive definiteness condition needed for the application of the results in the Appendix A 
of \cite{kmm}.  

The second modification comes from the eigenvalues.  We first note (\cite{ks}, \cite{k}) that the eigenvalues are simple, imaginary,
their total number  $2 N$ is finite and
$$
\aligned
N = \left [1/2 + {1 \over {\epsilon \pi}}||A||_{L^1}\right ]
\endaligned
$$
where $[.]$ here denotes the integer part of a real. \footnote{In \cite{ks} the exact estimate is stated in the abstract 
but a proof is only presented
for the case of $A$ with compact support. Still, the proof presented easily generalises for the case of non-compact support. In fact,
the crucial integral in (2.15) of \cite{ks} is positive $also$ in the non-compact support case 
(as is proved for example in  \cite{k}) and this in turn implies the exact estimate for the number of eigenvalues.}

As a result, the number of the actual eigenvalues is {\it equal} to the number of the ``Bohr-Sommerfeld-approximate 
eigenvalues", i.e. the {\it exact} roots of $e^{2iS(\mu)/\e}=-1$.

Now, our rigorous estimates in the previous section give a 1-1 correspondence between the eigenvalues 
$\lambda_n$ and the Bohr-Sommerfeld
approximations $\lambda_n^{\rm WKB}$, as long as  $|\lambda_n^{\rm WKB}|$ are 
greater than $  G\epsilon^{\alpha}$, for some constant $G>0$ independent of $\epsilon$ and 
for some $\alpha \in (0,1)$
and in fact $\lambda_n-\lambda_n^{\rm WKB}= o(\epsilon)$ uniformly in that set.

It follows that there is also a 1-1 correspondence between the rest of the eigenvalues 
$\lambda_n$ and the Bohr-Sommerfeld
approximations $\lambda_n^{\rm WKB}$, in which case  $\lambda_n^{\rm WKB}$ are in the closed disc of radius  
$  G\epsilon^{\alpha}$ and thus $\lambda_n$ have to be in a disc of radius of order $O(\epsilon^{\alpha})$ 
(possibly somewhat larger than $G\epsilon^{\alpha}$).
Clearly then $\lambda_n-\lambda_n^{\rm WKB}= O(\epsilon^{\alpha})$.

The crucial quantities to consider are the two ``Blaschke" products (see the first section of the appendix)
$$\left(\prod_{n=0}^{N-1}
\frac{\lambda-\lambda_n^*}{\lambda-\lambda_n}\right)$$
one for the actual eigenvalues, say $P$, and one for their WKB approximaton say $P_{\rm WKB}$.

Let the open disc $D^{\epsilon^{\alpha}}$ have center $0$ and radius $G\epsilon^{\alpha}.$
Outside the disc $D^{\epsilon^{\alpha}}$  
we have shown (in Corollary 2.9) that the difference between the actual eigenvalues
and their formal WKB approximation (which was used in \cite{kmm}) is of order $o(\epsilon)$ uniformly. 
Let $P^{\epsilon^{\alpha}}$ and $P_{\rm WKB}^{\epsilon^{\alpha}}$ be the ``Blaschke" products as above, but excluding
the eigenvalues $\lambda_n$ lying in $D^{\epsilon^{\alpha}} $.

A short calculation gives  $P^{\epsilon^{\alpha}}=P_{\rm WKB}^{\epsilon^{\alpha}} (1+o(1))$
uniformly in $\lambda  \in C \setminus D^{\epsilon^{\alpha}}$, since the number of terms in the product is
$O(\epsilon)$ and each ratio $\frac{\lambda_n-\lambda_n^{\rm WKB}}{\lambda}$ is $o(\epsilon)$.

A similar calculation relates the product over eigenvalues that lie in $D^{\epsilon^{\alpha}} $.
The corresponding ratio between
the Blaschke products  is $ (1+O(\epsilon^{2\alpha-1}))$. We then simply choose $\alpha > 1/2$.

By splitting each product into two products accordingly depending on  whether $|\lambda_n^{\rm WKB}|$ are
greater than $  G\epsilon^{\alpha}$ or not, we see that $P=P_{\rm WKB} (1+o(1))$.   
These estimates are pointwise, for any fixed $\lambda $ in the complex plane. But clearly they are also
uniform in any set consisting of the complement of a disc centered in 0
with radius  independent of $\epsilon$.

The rest of the argument is the same as for the reflection coefficients.
For $\lambda \in C \setminus  D $ where $D$ is the small but $\epsilon$-independent disc mentioned above
we have a uniform $o(1)$ approximation.
Inside  the small  set  $  D $ 
we have the right symmetry
and positive definiteness conditions. 
Then an appropriate
local parametrix exists inside $D$ according to the results in the Appendix A of \cite{kmm},
giving rise to an appropriate global approximative solution of the Riemann-Hilbert problem
 defined in terms of theta functions near $\lambda = \infty$ and thus
leading to the formulae stated in the introduction.

\end{proof}

\begin{remark}

It can happen (non-generically, for isolated values of $\epsilon$ ) that the reflection coefficent actually has a pole singularity at $0$.
In other words there is a $spectral~singularity$ at $0$. 
In such a case one can amend the analysis by considering a  very small  circle around $0$ say of radius $O(\epsilon),$ and removing the singularity
exactly in the same way we have removed the poles due to the eignevalues in \cite{kmm}. The reflection coefficient of course is not
analytically extensible in general but one can simply extract the singular part of the reflection coefficent which is of course rational.
The main result is not affected.
\end{remark}

In \cite{kr} we have studied the energy equilibrium problem that underlies the function $g$ appearing in the change of variables of
Chapter 4 in \cite{kmm} which expresses the finite gap ansatz.
We have been able to show that the equilibrium measure $\mu$ exists for a particular contour and hence that the right $g$ exists
so long as the support of  $\mu$ does not touch the segment $[0, iA]$ at more than a finite number of points. This is referred to as
Assumption (A) in \cite{kr}.

The next proposition follows.

\begin{proposition} 
Under the assumptions before the statement of Proposition 6.1, under assumption (A) in \cite{kr} 
and also the assumption that the density of eigenvalues admits an analytic extension in the
upper half-plane, the asymptotics \eqref{asymptotics} stated in the introduction are valid.
\end{proposition}

It has eventually become clear that both the finite gap ansatz stated in \cite{kmm} and the assumption (A) in \cite{kr}
are too restrictive. Hints of this inadequacy were already apparent in  \cite{kmm} and the phenomenon was further explored in  \cite{lm}.
  
In \cite{k2} we have added an amendment to \cite{kr} showing how to extend the analysis without the 
assumption (A). Also, in an unpublished preprint \cite{k1}, 
reproduced here in the last section of the appendix, we show
how to proceed if the assumption of analyticity for the eigenvalue density is not true, 
by solving an auxiliary scalar Riemann-Hilbert problem. We end up with the following result.

\begin{proposition} 
The asymptotics \eqref{asymptotics} stated in the introduction are always 
valid under initial data that satisfy the assumptions 
before the statement of Proposition 6.1 above.
\end{proposition}

ACKNOWLEDGEMENT.
\thanks{Research supported by the ARISTEIA II program of the Greek Secretariat of Research and Technology under Grant No.\ 3964.
The second author also acknowledges the generous support of Ritsumeikan University during three visits in 2015-2018.}

\appendix
\section{A Riemann-Hilbert
factorisation problem for the focusing nonlinear Schr\"odinger equation}

We first present some elementary facts about the Riemann-Hilbert 
factorisation formulation of the inverse scattering problem for the 
focusing nonlinear Schr\"odinger equation, as described in \cite{kmm}.
We describe first the case of reflectionless data which has been the main concern in  \cite{kmm}. 
Then we indicate how the problem changes  if we allow the reflection coefficient to be non-zero.

The focusing nonlinear Schr\"odinger equation is ``completely integrable".
Although there is no precise definition of this notion for infinite dimensional 
dynamical systems, one thing it always entails is the fact that it admits a ``Lax pair".  In our context
this means that, for arbitrary
$\epsilon$, it is represented as the compatibility condition
for two systems of linear ordinary differential equations:
\begin{equation}
\label{Lax1}
\epsilon\partial_x\left[\begin{array}{c}u_1\\u_2\end{array}\right]=
\left[\begin{array}{cc}-i\lambda & \psi\\-\psi^* & i\lambda\end{array}
\right]\left[\begin{array}{c}u_1\\u_2\end{array}\right]\,,
\end{equation}
\begin{equation}
\label{Lax2}
i\epsilon\partial_t\left[\begin{array}{c}u_1\\u_2\end{array}\right]=
\left[\begin{array}{cc} \lambda^2 - |\psi|^2/2 & i\lambda\psi -
\epsilon\partial_x\psi/2\\-i\lambda\psi^*-\epsilon\partial_x\psi^*/2 &
-\lambda^2 +|\psi|^2/2\end{array}\right]
\left[\begin{array}{c}u_1\\u_2\end{array}\right]\,,
\end{equation}
where $\lambda$ is an arbitrary complex parameter.

The $N$-soliton solutions  of the
nonlinear Schr\"odinger equation can be thought of as those complex functions
$\psi(x,t)$ for which there exist simultaneous column vector solutions
of the linear ODEs above in the particularly
simple form:
\begin{equation}
\begin{array}{rcl}
{\bf u}^+(x,t,\lambda)&=&\displaystyle\left[\begin{array}{c}
\displaystyle\sum_{p=0}^{N-1}A_p(x,t)\lambda^p\\\\
\displaystyle \lambda^N + \sum_{p=0}^{N-1}B_p(x,t)\lambda^p
\end{array}\right]\exp(i(\lambda x +\lambda^2 t)/\epsilon)\,,\\\\
{\bf u}^-(x,t,\lambda)&=&\displaystyle\left[\begin{array}{c}
\displaystyle\lambda^N +\sum_{p=0}^{N-1}C_p(x,t)\lambda^p\\\\
\displaystyle\sum_{p=0}^{N-1}D_p(x,t)\lambda^p\end{array}
\right]\exp(-i(\lambda x +\lambda^2 t)/\epsilon)\,,
\end{array}
\end{equation}
satisfying the relations
\begin{equation}
\begin{array}{rcl}
{\bf u}^+(x,t,\lambda_k) &=&\gamma_k{\bf u}^-(x,t,\lambda_k)\,,\\
-\gamma_k^*{\bf u}^+(x,t,\lambda_k^*)&=&{\bf u}^-(x,t,\lambda_k^*)\,,\hspace{0.2 in}k=1,\dots,N\,,
\end{array}
\end{equation}
for some distinct complex numbers $\lambda_0,\dots,\lambda_{N-1}$ in the
upper half-plane and nonzero complex numbers (not necessarily distinct)
$\gamma_0,\dots,\gamma_{N-1}$.  It is easy to check that given the numbers
$\{\lambda_k\}$ and $\{\gamma_k\}$, the relations
determine the coefficient functions $A_p(x,t)$, $B_p(x,t)$, $C_p(x,t)$
and $D_p(x,t)$ in terms of exponentials via the solution of a square
inhomogeneous linear algebraic system.  In the classic book of Faddeev and Takhtajan 
\cite{ft} it is shown that this linear system is always nonsingular assuming
the $\{\lambda_k\}$ are distinct and nonreal and that the $\{\gamma_k\}$
are nonzero.  The solution of the nonlinear Schr\"odinger equation for
which the column vectors ${\bf u}^\pm(x,t,\lambda)$ are simultaneous
solutions of the linear ODEs turns out to be
\begin{equation}
\psi(x,t)=2iA_{N-1}(x,t)\,.
\label{eq:psi}
\end{equation}

A typical initial condition $A(x)$  will not
correspond exactly to a multisoliton solution.  As is well-known (\cite{zs}, \cite{ft})
the procedure  generally
begins with the study the solutions of the linear ODEs for real
$\lambda$ and for $\psi=A(x)$.  One obtains from this analysis a
complex-valued {\em transmission coefficient} 
 $T(\lambda)=1/a(\lambda)$, $\lambda\in{\mathbb R}$. 
It turns out that the function $a(\lambda)$ has an
analytic continuation into the whole upper half-plane, and its zeros
occur at values of $\lambda$ for which  there is an
$L^2({\mathbb R})$ eigenfunction.  In this sense, the study of the
scattering problem for real $\lambda$ yields results for complex
$\lambda$ by unique analytic continuation.  The function $a(\lambda)$
can be interpreted as a Wronskian
between two particular solutions  that have
analytic continuations into the upper half-plane.  Thus at each $L^2$
eigenvalue $\lambda_k$, there is a complex number $\gamma_k$
that is the ratio of these two analytic solutions.  In addition to the
transmission coefficient, one also finds a complex-valued function
$b(\lambda)$ that gives rise to a {\em reflection coefficient}
$R(\lambda):=b(\lambda)/a(\lambda)$,
$\lambda\in{\mathbb R}$.  Following
 Zakharov and Shabat \cite{zs} we have:
\begin{enumerate}
\item
When $\psi(x,t)$ is the solution of the focusing NLS  with initial data
$A(x)$, then for each $t>0$ one has different coefficients in the
linear problem, and therefore the eigenvalues
 $\{\lambda_k\}$, proportionality constants
 $\{\gamma_k\}$ and the
function $b(\lambda)$, can be computed independently for each $t>0$.
However, the eigenvalues
$\{\lambda_k\}$ (more generally the function $a(\lambda)$) and also
$|b(\lambda)|$, $\lambda\in{\mathbb R}$, are independent of $t$, and
the proportionality constants $\{\gamma_k\}$ and $\arg(b(\lambda))$,
$\lambda\in{\mathbb R}$ evolve simply in time.  Thus,
$R(\lambda,t)=R(\lambda,0)\exp(-2i\lambda^2 t/\epsilon)$ and
$\gamma_k(t)=\gamma_k(0)\exp(-2i\lambda_k^2 t/\epsilon)$.
\item
The function $\psi(x,t)$ can be reconstructed at later times $t>0$ in
terms of the discrete spectrum $\{\lambda_k\}$, $\{\gamma_k\}$, and
the reflection coefficient $R(\lambda)$.
\end{enumerate}

If for the initial condition $A(x)$ we have
$b(\lambda)\equiv 0$, then the step of reconstructing the solution of
the initial value problem  is essentially what we have
already described.  Namely, one solves the linear equations
 for the coefficient $A_{N-1}(x,t)$ and then the
solution of the initial value problem is given by (\ref{eq:psi}).  $N$
turns out to be the number of $L^2$ eigenvalues for $\psi_0(x)$ in the upper
half-plane.

In general, the reconstruction of $\psi$ from the scattering data can
be recast in terms of the solution of a matrix-valued meromorphic
Riemann-Hilbert problem.
One seeks (for each $x$ and $t$, which play the role of parameters) a
matrix-valued function ${\bf m}(\lambda)$ of $\lambda$ that is jointly
meromorphic in the upper and lower half-planes and for which
\begin{enumerate}
\item
${\bf m}(\lambda)\rightarrow {\mathbb I}$ in each half-plane as
$\lambda\rightarrow \infty$.
\item
The singularities of ${\bf m}(\lambda)$ are completely specified.
There are simple poles at the eigenvalues $\{\lambda_k\}$ and the
complex conjugates with residues of a certain specified type (see
below).
\item
On the real axis $\lambda\in{\mathbb R}$, there is the {\em jump
relation}
\begin{equation}
{\bf m}_+(\lambda)=
{\bf m}_-(\lambda){\bf v}(\lambda)\,,
\hspace{0.2 in}
{\bf m}_\pm(\lambda):=
\lim_{\eta\downarrow 0}{\bf m}(\lambda\pm i\eta)
\label{eq:updowndef}
\end{equation}
where ${\bf v}(\lambda)$ is a certain {\em jump matrix} \index{jump
matrix} built out of $R(\lambda)$ and depending {\em explicitly} on
$x$ and $t$ (and $\epsilon$).  The jump matrix becomes the identity matrix for
$b(\lambda)\equiv 0$.
\end{enumerate}

If the boundary values ${\bf m}_\pm(\lambda)$ are continuous, and
if $b(\lambda)\equiv 0$, then it is easy to see that the solution ${\bf
m}(\lambda)$ must be a rational function of $\lambda$.  
In \cite{kmm}  this is the only case considered. In the current paper however
the jump matrix is non-trivial. In fact
\begin{equation}
{\bf v}(\lambda)=
\left[\begin{array}{cc}1 & R(\lambda) \exp\left(\frac{1}{\epsilon}(-2i\lambda x-2i\lambda^2 t)\right)\\
R^* (\lambda) \exp\left(\frac{1}{\epsilon}(2i\lambda x+2i\lambda^2 t)\right) & 1+|R(\lambda)|^2 \end{array}\right]
\label{eq:jumpreal}
\end{equation}

Continuing with the pure soliton case of $b(\lambda)\equiv 0$, 
from the column vectors
${\bf u}^\pm(x,t,\lambda)$, we build a matrix solution of \eqref{Lax1}-\eqref{Lax2}:
\begin{equation}
\Psi(\lambda):=
[{\bf u}^-(x,t,\lambda),{\bf u}^+(x,t,\lambda)]
\mbox{diag}\left(\prod_{j=1}^N(\lambda-\lambda_j)^{-1},
\prod_{j=1}^N(\lambda-\lambda_j^*)^{-1}\right)
\exp(i\sigma_3\lambda^2 t/\epsilon)\,.
\end{equation}
This special matrix solution is the familiar
Jost solution.
If we now define a
matrix ${\bf m}(\lambda)$ by
\begin{equation}
{\bf m}(\lambda):= \Psi(\lambda)\exp(i\sigma_3\lambda x/\epsilon)\,,
\label{eq:Psi-to-m}
\end{equation}
then we find using  \eqref{Lax1}-\eqref{Lax2}
that for all fixed complex $\lambda$
different from the eigenvalues $\{\lambda_k\}$ and their complex
conjugates, ${\bf m}(\lambda)$ is a uniformly bounded function of $x$
that satisfies ${\bf m}(\lambda)\rightarrow{\mathbb I}$ as
$x\rightarrow +\infty$.

We can
deduce from the explicit form of the vectors ${\bf
u}^\pm(x,t,\lambda)$ and from the relations \eqref{Lax1}-\eqref{Lax2}
that ${\bf m}(\lambda)$ solves the following problem.

Given the discrete data $\{\lambda_k\}$ and $\{\gamma_k\}$, find a
matrix ${\bf m}(\lambda)$ with the following two properties:
\begin{enumerate}
\item
${\bf m}(\lambda)$ is a rational function of $\lambda$, with simple
poles confined to the eigenvalues $\{\lambda_k\}$ and the complex conjugates.
At the singularities:
\begin{equation}
\begin{array}{rcl}
\displaystyle
\mathop{\rm Res}_{\lambda=\lambda_k}
{\bf m}(\lambda)&=&
\displaystyle\lim_{\lambda\rightarrow\lambda_k}{\bf m}(\lambda)
\left[\begin{array}{cc}0 & 0\\c_k(x,t) & 0\end{array}\right]
\,,\\\\
\displaystyle
\mathop{\rm Res}_{\lambda=\lambda_k^*}
{\bf m}(\lambda)&=&
\displaystyle\lim_{\lambda\rightarrow\lambda_k^*}{\bf m}(\lambda)
\left[\begin{array}{cc}0 & -c_k(x,t)^*\\0 & 0\end{array}\right]
\,,
\end{array}
\label{eq:residueconds}
\end{equation}
for $k=0,\dots,N-1$, with
\begin{equation}
c_k(x,t):=\left(\frac{1}{\gamma_k}\right)
\frac{\displaystyle\prod_{n=0}^{N-1}(\lambda_k-\lambda_n^*)}
{\displaystyle\prod_{\scriptstyle n=0\atop
\scriptstyle n\neq k}^{N-1}(\lambda_k-\lambda_n)}
\exp(2i(\lambda_k x +\lambda_k^2
t)/\epsilon)\,.
\label{eq:residuecoeffs}
\end{equation}
\item
\begin{equation}
{\bf m}(\lambda)
\rightarrow{\mathbb I}\,,\hspace{0.2 in}\mbox{as}\hspace{0.2 in}
\lambda\rightarrow\infty\,.
\end{equation}
\end{enumerate}

These two properties actually characterise the matrix function ${\bf
m}(\lambda)$ uniquely.
We have (\cite{kmm})

\begin{proposition}
The meromorphic Riemann-Hilbert Problem corresponding to
the discrete data $\{\lambda_k\}$ and $\{\gamma_k\}$ has a unique
solution whenever the $\lambda_k$ are distinct in the upper half-plane
and the $\gamma_k$ are nonzero.  The function defined from the
solution by
\begin{equation}
\psi:=2i\lim_{\lambda\rightarrow\infty}\lambda m_{12}(\lambda)
\label{eq:psiagain}
\end{equation}
(that this limit exists is part of the proposition) is a nontrivial
$N$-soliton solution of the focusing nonlinear Schr\"odinger
equation.
\end{proposition}

For an asymptotic analysis it is useful
to convert the meromorphic Riemann-Hilbert problem back
into a sectionally holomorphic Riemann-Hilbert
problem. This can be easily be done by  constructing (for example) small circles around the poles and redifining
the unkonwn inside those circles accordingly, see \cite{DKKZ}.
Here, we proceed as follows.
 
Let $C$ be a simple closed contour that is the boundary
of a simply-connected domain $D$ in the upper half-plane that contains
$all$ of the eigenvalues $\{\lambda_k\}$.  We assign to $C$ a counterclockwise 
orientation.
By $C^*$ and
$D^*$ we mean the corresponding complex conjugate sets in the lower
half-plane, and we assign  both
loops  the same orientation.

It is not hard to see (\cite{kmm}) that for our symmetric even data $A(x)$ one has $\gamma_k=(-1)^k$.
Still, it has  proved   convenient in the asymptotic analysis of \cite{kmm} and \cite{kr}
to interpolate the proportionality constants 
as follows. One can easily
choose a constant $Q$ (always 1 or -1, but depending on $x,t$) and a function $X(\lambda)$
analytic in $D$ so that
\begin{equation}
\gamma_k=Q\exp(X(\lambda_k)/\epsilon)\,,\hspace{0.3 in}k=0,\dots,N-1\,.
\end{equation}
In general, $X(\lambda)$ could be systematically constructed as an
interpolating polynomial of degree $\sim N$.  In our (symmetric) case
the phases $\gamma_k$ are highly correlated so that for
very large $N$ one can easily choose for $X(\lambda)$ a polynomial of low
degree or another simple expression. Note that the
interpolant of the $\gamma_k$ is not
necessarily unique; for each $K$ in some indexing set (an integer) there is a
distinct pair $(Q_K,X_K(\lambda))$ such that for all $k$,
$\gamma_k=Q_K\exp(X_K(\lambda_j)/\epsilon)$. 

\begin{remark} 
In \cite{k2} we have made use of this
freedom.  We have found that the best choice depends on the Riemann surface sheet where our contour is allowed to expand. 
The issue of the right choice is also related to the improvement of the approximation achieved in \cite{lm}.

\end{remark}

With the help of the interpolant of the proportionality
constants\index{proportionality constants!interpolant of}, we define a
new matrix ${\bf M}(\lambda)$ for $\lambda\in {\mathbb C}\setminus
(C\cup C^*)$ in the following way.  First, for all $\lambda\in D$, set
\begin{equation}
{\bf M}(\lambda):={\bf m}(\lambda)
\left[\begin{array}{cc}
1 & 0\\\\
\displaystyle -\left(\frac{1}{Q_K}\right)\left(\prod_{n=0}^{N-1}
\frac{\lambda-\lambda_n^*}{\lambda-\lambda_n}\right)\exp\left(\frac{1}{\epsilon}
(2i\lambda x+
2i\lambda^2 t -X_K(\lambda))\right) & 1\end{array}\right]
\,.
\label{eq:Blaschke}
\end{equation}
Next, for all $\lambda\in D^*$, set
\begin{equation}
{\bf M}(\lambda):=\sigma_2{\bf M}(\lambda^*)^*\sigma_2\,.
\label{eq:under}
\end{equation}
Finally, for all $\lambda\in {\mathbb C}\setminus (\overline{D}\cup
\overline{D}^*)$ ({\em i.e.} in the rest of the complex plane minus
$C\cup C^*$) simply set
\begin{equation}
{\bf M}(\lambda):={\bf m}(\lambda)\,.
\label{eq:outside}
\end{equation}

It is straightforward to verify that by our choice of interpolants,
and the ``Blaschke" factor  appearing in
(\ref{eq:Blaschke}), that ${\bf M}(\lambda)$ {\em has no poles in $D$
or $D^*$ and hence is sectionally holomorphic in the complex $\lambda$
plane}.  By definition, we have preserved the reflection symmetry of
${\bf m}(\lambda)$ so that for all $\lambda\in {\mathbb C}
\setminus (C\cup C^*)$ we have:
\begin{equation}
{\bf M}(\lambda^*)=\sigma_2{\bf M}(\lambda)^*\sigma_2\,.
\end{equation}
The matrix ${\bf M}(\lambda)$ has continuous boundary values from
either side on $C$ and $C^*$.  To describe these, let the left
(respectively right) side of the oriented contour $C\cup C^*$ be
denoted by ``$+$'' (respectively ``$-$'').  For $\lambda\in C\cup
C^*$ define
\begin{equation}
{\bf M}_\pm(\lambda):=\!\!\!\!\!\!\!\!\!\!\!\!\!\!\!\!
\lim_{\begin{array}{c}\scriptstyle\mu\rightarrow\lambda
\\\scriptstyle\mu\in\pm\mbox{ side of }C\cup C^*\end{array}}
\!\!\!\!\!\!\!\!\!\!\!\!\!\!\!\!{\bf M}(\mu)\,,
\end{equation}
that is, the nontangential limits from the left and right sides.
Then, using the fact that ${\bf m}(\lambda)$ is analytic on $C\cup
C^*$ and the piecewise definition of ${\bf M}(\lambda)$ given by
(\ref{eq:Blaschke}), (\ref{eq:under}), and (\ref{eq:outside}),
we find
\begin{equation}
\begin{array}{rcll}
{\bf M}_+(\lambda)&=&{\bf M}_-(\lambda){\bf
v}_{\bf M}(\lambda)\,,&\lambda\in C_\omega\,,\\\\
{\bf M}_+(\lambda)&=&{\bf M}_-(\lambda)
\sigma_2{\bf v}_{\bf M}(\lambda^*)^*\sigma_2\,,
&\lambda\in C^*\,,
\end{array}
\label{eq:jumps}
\end{equation}
where for $\lambda\in C$,
\begin{equation}
{\bf
v}_{\bf M}(\lambda):=\left[\begin{array}{cc}
1 & 0\\\\ \displaystyle
-\omega\left(\frac{1}{Q_K}\right)\left(\prod_{n=0}^{N-1}
\frac{\lambda-\lambda_n^*}{\lambda-\lambda_n}\right)
\exp\left(\frac{1}{\epsilon}(2i\lambda x+2i\lambda^2 t
-X_K(\lambda))\right) & 1
\end{array}\right]\,.
\label{eq:vM}
\end{equation}
Now, by defining the discrete measure
\begin{equation}
d\mu = \sum_{k=0}^{N-1} \left[\epsilon\delta_{\lambda_k^*} -
\epsilon\delta_{\lambda_k}
\right]\,,
\label{eq:discretetrue}
\end{equation}
we see that for any branch of the logarithm,
\begin{equation}
\prod_{k=0}^{N-1}\frac{\lambda-\lambda_k^*}{\lambda-\lambda_k} = \exp\left(
\frac{1}{\epsilon}\int \log(\lambda-\eta)\,d\mu(\eta)\right)\,.
\label{eq:Blaschke2}
\end{equation}

In the general case (with a nontrivial reflection coefficient $R$) the same calculation
applies and does not affect the jump across the real line.
Suppose then the eigenvalues $\{\lambda_k\}$ and proportionality constants
$\{\gamma_k\}$ are given along with an appropriate interpolation
$Q_K\exp(X_K(\lambda)/\epsilon)$ of the $\gamma_k$ and a smooth closed 
oriented contour $C$ enclosing the eigenvalues in the upper half-plane. 
Suppose also that $R$, the reflection coefficient corresponding to the initial data
is also given. We
define a Riemann-Hilbert problem as follows.
 
Find a matrix function ${\bf M}(\lambda)$ that satisfies:
\begin{enumerate}
\item
 ${\bf M}(\lambda)$ is analytic in each component of
${\mathbb C}\setminus (C\cup C^*)$.
\item
 ${\bf M}(\lambda)$ assumes continuous boundary
values on $C\cup C^*$ and the real line.
\item
 The boundary values taken on $C\cup C^*$
satisfy the relations (\ref{eq:jumps}) with ${\bf v}_{\bf M}(\lambda)$
given explicitly by (\ref{eq:vM}).
In the case where the reflection coefficient is non-trivial, there is  
a jump condition across the real line (\ref{eq:jumpreal}).
\item
 ${\bf M}(\lambda)$ is normalized at infinity:
\begin{equation}
{\bf M}(\lambda)\rightarrow {\mathbb I}\mbox{  as  }
\lambda\rightarrow\infty\,.
\end{equation}
\end{enumerate}

\begin{proposition}
The holomorphic Riemann-Hilbert Problem~ has a unique
solution ${\bf M}(\lambda)$ whenever the $\lambda_k$ are distinct and
nonreal, and the $\gamma_k$ are nonzero.  The function defined by
\begin{equation}
\psi :=2i\lim_{\lambda\rightarrow\infty}\lambda M_{12}(\lambda)\,,
\label{eq:fieldatinfinity}
\end{equation}
is independent of the 
particular choice of loop contour $C$ and interpolant index $K$,
and is the  solution of the focusing nonlinear Schr\"odinger
equation corresponding to the reflection coefficient $R$ and the discrete data $\{\lambda_k\}$ and
$\{\gamma_k\}$.
\end{proposition}

It is possible to allow $C$ to meet the
real axis at one or more isolated points $u_k\in{\mathbb R}$, as long
as at each $u_k$ the incoming and outgoing parts of $C$ make nonzero
angles with the real axis and with each other.  The contour $C$ should
thus meet the axis in ``corners'' (if at all).

\section{Asymptotic analysis via a Riemann-Hilbert deformation } 

In view of \ref{eq:vM} and \ref{eq:Blaschke2},
the Riemann-Hilbert problem we have to analyse asymptotically can be seen as 
a (nonlinear) analogue of exponential integrals.
While in linear problems where the Fourier integral method can be applied we end up with exponential integrals,
here we have a $Riemann-Hilbert~problem$  with exponential phase.

It was first realized back in 1981 by Alexander Its (\cite{its})
that the long time asymptotics for the solution of the initial value problem
to the focusing NLS 
can be extracted by reducing the Riemann-Hilbert problem  to a
``model" Riemann-Hilbert problem that can be solved explicitly, exactly
as one does in the asymptotic analysis of exponential integrals.
(Apparently Its was inspired by work of Jimbo, Miwa and Ueno on the isomonodromy
method for Painleve equations \cite{jmu}.)
The Riemann-Hilbert problem deformation
method has been made rigorous
and systematic in later work by Deift and Zhou in 1993 (\cite{dz}).
The  basic ideas of the Deift-Zhou method are:

1. Equivalence of the solvability of a matrix Riemann-Hilbert problem to
the invertibility of an associated singular integral operator.
Expression of the solution of the matrix Riemann-Hilbert problem
as a singular integral involving the inverse of the associated singular integral operator.

The idea goes back to the Georgian school of Mushkelishvilli  and provides a nice way to show that
under some conditions, small changes in the jump data result in small changes in the solution.

2. Appropriate lower/diagonal/upper factorisations of jump matrices.

3. Introduction and solution of auxiliary scalar problems leading to a conjugation
of the original problem by an exponential factor (the ``g-function"). 

The semiclassical problem is more complicated and requires two more ideas.

4. An auxiliary variational problem of ``electrostatic type" (going back to work of Lax, Levermore and Venakides in 
the 1980s on the zero dispersion KdV problem). 
Solution of the Euler-Lagrange equations for this problem via theta functions.

5. Search for an optimal contour (selection of a contour of steepest descent) where all of the above can be applied.
Deformation from one contour to another. This is a feature appearing only in problems with
non-self-adjoint Lax operator, where the spectrum is not necessarily real and the whole
deformation procedure is conducted fully in the complex plane.

A more detailed expedition of the above ideas appears in \cite{k3}.
The rigorous implementation of the whole sequence is done in \cite{kmm} and \cite{kr}.

Our first step is to employ a ``change of variables"
$N(z)=M(z) \exp({{g(z) \sigma_3} \over {\epsilon}} ) $
which will enable us to reduce the given Riemann-Hilbert problem to one that is easier to handle
and which asymptotically will be explicitly solvable.

The function $g$  is  given as a logarithmic transform of an ``equilibrium" measure on a particular smooth curve:
$g(z) = \int \log(z-\eta) d \mu (\eta)$ where $d \mu$ is the equilibrium measure corresponding to a particular external field
depending on the parameters $x,t$ and the initial data. The particular curve on which the equilibrium problem is defined 
is chosen such that it maximises the corresponding equilibrium energy.
In other words the equilibrium measure 
solves a max-min type variational problem.

More precisely,
let $C $ be  a contour enclosing the eigenvalues in the upper half-plane as described in the previous section.
A priori we seek a function satisfying
$$
\aligned
g(\lambda) \text{ is independent of } \hbar.\\
g(\lambda) \text{ is analytic for }
\lambda\in \Bbb C \setminus (C\cup C^*).\\
g(\lambda)\rightarrow 0 \text{ as }
\lambda\rightarrow\infty.\\
g(\lambda) \text{ assumes 
continuous boundary
values from both sides of }C\cup C^*,\\
\text{ denoted by } g_+(g_-) \text{ on the left (right) of }C \cup C^*.\\
g(\lambda^*)+g(\lambda)^* = 0 \text{ for
all }\lambda\in \Bbb C\setminus (C\cup C^*).
\endaligned
$$

The assumptions above permit us to write $g$ in terms of a measure $\rho$
defined on the
contour $C \cup C^*$. Indeed
$$
\aligned
g(\lambda) =  \int_{C \cup C^*} \log(\lambda -\eta) \rho(\eta) d\eta,
\endaligned
$$
for an appropriate definition of the logarithm branch.

Further  technical  conditions are necessary to ensure that  the Riemann-Hilbert
deformations required can  go through. Such conditions characterise $g~~~and$ the contour $C$.
In \cite{kmm} these are  given by conditions (4.20) and (4.31). 

In \cite{kr} we define $g$ in a somewhat different but equivalent way, in terms of an equilibrium energy
problem. For any $given$ contour $C$ we choose $\rho(\eta) d\eta$ to be the equilibrium measure 
for a certain external field depending on $x,t$ and $A(x)$; but eventually we choose $C$ that maximizes the
equilibrium energy.
The extra  technical  conditions on $g$ are equivalent to the Euler-Lagrange conditions for the
energy equilibrium problem.

The Riemann-Hilbert problem is then asymptotically ``deformed" to a ``model" problem that can be explicitly solved 
in terms of theta functions. The model problem is an asymptotic semiclassical approximation of the original one.
The semiclassical asymptotics of the focusing NLS problem are thus also recovered via \ref{eq:fieldatinfinity}.

\section{On the Analyticity of the Spectral Density}

It is essential for the proofs in \cite{kmm}  that the ``density
of eigenvalues" $\rho^0(\eta)$ (see (3.2) of  \cite{kmm}),
derived by WKB theory and a priori
defined in the straight line interval connecting
$0$ to $iA$, be analytically extensible
to the closed upper half-plane  $\Bbb H$.
The main issue is whether the function
$$
\aligned
R^0(\eta) = \int^{x_+(\eta)}_{x_-(\eta)} (A(x)^2 +\eta^2)^{1/2} dx,
\endaligned
$$
where the  points $x_{\pm}$ are defined by
$$
\aligned
A( x_{\pm} (\eta)) = -i\eta, ~~~0< -i\eta < A_0, \\
-A_0 < x_- (\eta) < 0 < x_+ (\eta) < A_0,
\endaligned
$$
admits an analytic extension.
We note here that we choose the branch of the square root
that is positive for $x_- < x <x_+$.

We will show that even if $R^0$ does not admit an  analytic
extension  in $\Bbb H $, the analysis of Chapter 5 in  \cite{kmm}
can be amended via the solution of  a scalar Riemann-Hilbert problem.

Indeed, consider the following scalar additive Riemann-Hilbert problem,
with jump on the linear segment $\Sigma = [-iA_0, iA_0]$. Let $p$ be a function
analytic in $ \Bbb C \setminus [-iA_0, iA_0]$, such that
$$
\aligned
p_+ (\eta) + p_-(\eta) =  \rho_0(\eta)= {{dR^0} \over {d\eta}}, ~~~\eta \in (-iA_0, iA_0). 
\endaligned
$$
Indeed, let 
$$
\aligned
p (\eta) = (A_0^2+ \eta^2 )^{1/2}  \int_{(-iA_0, iA_0)} \frac{\rho_0(s)}{ (A_0^2+ s^2 )^{1/2}} \frac{ds}{2 \pi i (s-\eta)}, 
~~~\eta \in \Bbb C \setminus (-iA_0, iA_0).  \endaligned 
$$
Here $R^0(\eta)$ is extended to the lower half of $\Sigma$
by the relation $R^0(\eta^*) = R^0(\eta)$.
The ``+" side is to the left of $\Sigma$ and the
``-" side is to the right of $\Sigma$.

Note that if
$R^0$ is entire, then we can choose $p = \rho^0 = 1/2 {{dR^0} \over {d\eta}}.$
In general, our choice of initial data only ensures that $\rho^0$
is continuous.

Now, the analysis of Chapter 5 in  \cite{kmm} can be amended as follows.
First, let's amend the definition of $X$ in Chapter 3, which describes
the interpolant of the norming constants.
We simply set
$$
\aligned
X(\lambda) = i \pi (2K+1)  \int^{iA_0}_{\lambda} (p_+(\eta) +p_-(\eta) ) d\eta,
\endaligned
$$
for $\lambda $ in the linear segment $[0, iA_0]$.
Then, the discussion of Chapter 5 in  \cite{kmm}, in particular from relation (5.4)
to (5.8), is  amended by substitutitng $\bar \rho^{\sigma} = p-\rho$.
More precisely, taking $\sigma = 1$,
$$
\aligned
\int_{0}^{iA_0}L^0_\eta(\lambda)p_-(\eta)d\eta =
\int_{C_I}
L^{C}_{\eta-}(\lambda)p(\eta)d\eta,
\endaligned
$$
and similarly, by symmetry,
$$
\aligned
\int_{-iA_0}^{0}L^0_\eta(\lambda)p_-(\eta^*)^*\,d\eta =
\int_{C_I^*}
L^{C}_{\eta-}(\lambda)p(\eta^*)^*\,d\eta.
\endaligned
$$
(Recall here that $L^0_\eta(\lambda)=log (\lambda-\eta),$
with a cut along the imaginary axis from $\eta$ to $-i \infty$.
In the above integral we integrate over the ``-" side, while in the integral just following
we integrate over the ``+" side.)
Also
$$
\aligned
\int_{0}^{iA_0}L^0_\eta(\lambda)p_+(\eta)d\eta =
\int_{C_F}
L^{C}_{\eta-}(\lambda)p(\eta)d\eta,
\endaligned
$$
and similarly, by symmetry,
$$
\aligned
\int_{-iA_0}^{0}L^0_\eta(\lambda)p_+(\eta^*)^*d\eta =
\int_{C_F^*}
L^{C}_{\eta-}(\lambda)p(\eta^*)^*d\eta.
\endaligned
$$

Next, note that
$L^{C}_{\eta+}(\lambda)=L^{C}_{\eta-}(\lambda)$ for all
$\eta\in C_I\cup C_I^*$ ``below'' $\lambda\in C_I$
and at the same time
$L^{C}_{\eta_+}(\lambda)= 2\pi i +
L^{C}_{\eta-}(\lambda)$ for $\eta\in C_I$ ``above'' $\lambda$.
This means that for $\lambda\in C$,
$$
\aligned
\int_{C} L^{C}_{\eta\pm}(\lambda)p(\eta)d\eta +
\int_{C^*} L^{C}_{\eta\pm}(\lambda)p(\eta^*)^*d\eta =
\\
\int_{C} \overline{L^{C}_\eta}(\lambda)p(\eta)d\eta
+ \int_{C^*}\overline{L^{C}_\eta}(\lambda)p(\eta^*)^* d\eta
\pm\pi i/2 \int_{C_I}p(\eta)d\eta
\pm\pi i/2 \int_{C_F}p(\eta)d\eta,
\endaligned
$$
with $\overline{L^{C}_\eta}(\lambda) = {{L^{C}_{\eta+}(\lambda)
+L^{C}_{\eta-}(\lambda) } \over 2}.$
Assembling these
results gives the expression
$$
\aligned
\tilde{\phi}(\lambda)=
\int_{C}\overline{L^{C}_\eta}(\lambda)\overline{\rho}(\eta) d\eta +
\int_{C^*}\overline{L^{C}_\eta}(\lambda)\overline{\rho}(\eta^*)^*d\eta\\
+J(2i\lambda x + 2i\lambda^2 t) -(J(2K+1)+1) ~
(\pm\pi i/2 \int_{C_I}p(\eta)d\eta
\pm\pi i/2 \int_{C_F}p(\eta)d\eta),
\endaligned
$$
valid for $\lambda \in C$, where we have introduced the
complementary density  for
$\eta\in C:
\overline{\rho}(\eta):=p(\eta)-\rho(\eta).  $
Choosing $K$ so that $J(2K+1)+1 =0$, the last term vanishes and
we simply have
$$
\aligned
\tilde{\phi}(\lambda)=\int_{C}\overline{L_\eta^{C,\sigma}}(\lambda)
\overline{\rho}(\eta)d\eta +
\int_{C^*}\overline{L_\eta^{C}}(\lambda)
\overline{\rho}(\eta^*)^*d\eta +
J(2i\lambda x + 2i\lambda^2 t).
\endaligned
$$
Comparing with (5.11) of  \cite{kmm} this last formula  is less awkward,
since it does not depend on the a priori constraint
that  the contour $C$ has to go through $iA$,
a constraint that is eventually suspended anyway.

The rest of the proofs of  \cite{kmm} go through, with $p$ substituting $\rho^0$.
We omit the detailed discussion, but
we $do$ stress one major point on the variational problem
of Chapter 8 of  \cite{kmm}.

The contour $C$  and the measure $\rho d\eta$
are characterized by a solution of a Green's  variational
problem of electrostatic kind. Indeed
$$
\aligned
E_{\phi} (\rho d\eta) = \max_{C'} \min_{\mu: {\rm supp}(\mu)  \in C} E_{\phi} (\mu),
\endaligned
$$
where the contours $C'$   are a priori supported in the upper
half-plane minus the linear segment $[0, iA_0]$, and
$E_{\phi}$ is the weighted energy of a measure with respect to the external
field  given by
$$
\aligned
\phi (z) =
\int \log {{ |z-\eta^*| } \over {|z-\eta|}}
\rho^0 (\eta) d\eta - \re (i\pi J \int_{z}^{iA_0} p (\eta) d\eta
+2iJ (z x + z^2 t) ).
\endaligned
$$
The harmonicity of $\phi$ is important to the structure of
$C, supp(\rho)$. But again, even if $\rho^0$ is not analytically extended,
it can be written as a sum of two terms that $are$.

One could write $\phi$ as
$$
\aligned
\phi (z) =
\int \log {{ |z-\eta^*| } \over {|z-\eta|}}
(p_+ + p_-) (\eta) d\eta - \re (i\pi J \int_{z}^{iA_0} p (\eta) d\eta
+2iJ (z x + z^2 t) ).
\endaligned
$$
Again, this representation is perhaps more natural, since in setting
the variational problem it is more appropriate to think of the
``left" and ``right" sides of the linear segment
$[0, iA_0]$ as distinct.

\begin{remark}

The moral of the story is that if $\rho^0$ does not admit an entire extension,
we can write it as the average of two functions $p_-, p_+$ that can be extended
to the left and right of the segment $[0, iA_0]$ respectively,
and proceed as before, with $\rho^0$  substituted by $p$.

\end{remark}

\begin{remark}
In \cite{kr}  we assume that the solution of the variational problem does not touch the
spike $[0, iA_0]$ except possibly at a finite number of points. As shown in \cite{k2}, this obstacle can be
overcome by setting the variational problem on an infinite sheeted Riemann surface $\Bbb L$. For this,
we use the analyticity of $\rho^0$ even across the spike. Here we don't have that (in fact this is
the whole point of this appendix). But a careful examination of \cite{k2} shows that what we actually need is
analyticity across all but one liftings of the spike on $\Bbb L$. This we can get by simply setting our
scalar Riemann-Hilbert problem on $\Bbb L$ and letting the jump be a single copy of the spike $[0, iA]$
in $\Bbb L$. The scalar Riemann-Hilbert problem on $\Bbb L$ can be explicitly solved by mapping conformally
$\Bbb L$ to $\Bbb C$.

\end{remark}

\end{document}